\documentclass[acmtopc]{acmsmallcustom}

\usepackage{amsfonts}
\usepackage{amsmath}
\usepackage{mathabx}
\usepackage{comment}
\usepackage{stmaryrd}

\usepackage[ruled]{algorithm2e}

\SetAlFnt{\small}
\SetAlCapFnt{\small}
\SetAlCapNameFnt{\small}
\SetAlCapHSkip{0pt}
\IncMargin{-\parindent}

\acmVolume{9}
\acmNumber{4}
\acmArticle{1}
\acmYear{2010}
\acmMonth{3}

\begin{document}

\markboth{H. Chen}{The Tractability Frontier of Graph-Like First-Order Query Sets}

\title {The Tractability Frontier of Graph-Like First-Order Query Sets}
\author{HUBIE CHEN
\affil{Universidad del Pa\'{i}s Vasco,
E-20018 San Sebasti\'{a}n, Spain; and IKERBASQUE, Basque Foundation
for Science,
E-48011 Bilbao
}
}

\begin{abstract}
{\bf Abstract.}
We study first-order model checking, by which we refer to the
problem of deciding
whether or not a given first-order sentence is satisfied by a given
finite structure.
In particular, we aim to understand on which sets of sentences
this problem is tractable, in the sense of parameterized complexity
theory.
To this end, we define the notion of a graph-like sentence set,
which definition is inspired by previous work on first-order model
checking 
wherein the permitted connectives and quantifiers were restricted.
Our main theorem is the complete tractability classification of such graph-like
sentence sets, which is (to our knowledge) the first complexity classification theorem
concerning a class of sentences that has no restriction on the
connectives
and quantifiers.
To present and prove our classification, we introduce and develop 
a novel complexity-theoretic framework which is built on
parameterized complexity and includes new notions of reduction.
\end{abstract}

\maketitle

\newcommand{\footbox}{\ensuremath{\footnotesize \textup{$\Box$}}}

\newtheorem{prop}[theorem]{Proposition}
\newtheorem{assumption}[theorem]{Assumption}

\newcommand{\ppequiv}{\mathsf{PPEQ}}
\newcommand{\eq}{\mathsf{EQ}}
\newcommand{\iso}{\mathsf{ISO}}
\newcommand{\ppeq}{\ppequiv}
\newcommand{\ppiso}{\mathsf{PPISO}}
\newcommand{\boolppiso}{\mathsf{BOOL}\mbox{-}\mathsf{PPISO}}
\newcommand{\csp}{\mathsf{CSP}}
\newcommand{\gi}{\mathsf{GI}}
\newcommand{\ci}{\mathsf{CI}}
\newcommand{\rela}{\mathbf{A}}
\newcommand{\relb}{\mathbf{B}}
\newcommand{\relc}{\mathbf{C}}
\newcommand{\alga}{\mathbb{A}}
\newcommand{\algb}{\mathbb{B}}
\newcommand{\algab}{\mathbb{A}_{\relb}}

\newcommand{\idemp}{I}

\newcommand{\varv}{\mathcal{V}}
\newcommand{\variety}{\mathcal{V}}
\newcommand{\false}{\mathsf{false}}
\newcommand{\true}{\mathsf{true}}
\newcommand{\pol}{\mathsf{Pol}}
\newcommand{\inv}{\mathsf{Inv}}
\newcommand{\cc}{\mathcal{C}}
\newcommand{\alg}{\mathsf{Alg}}
\newcommand{\pitwo}{\Pi_2^p}
\newcommand{\sigmatwo}{\Sigma_2^p}
\newcommand{\pithree}{\Pi_3^p}
\newcommand{\sigmathree}{\Sigma_3^p}

\newcommand{\fancyg}{\mathcal{G}}
\newcommand{\tw}{\mathsf{tw}}

\newcommand{\mc}{\mathsf{MC}}
\newcommand{\mcs}{\mathsf{MC}_s}

\newcommand{\mcb}{\mathsf{MC_b}}

\newcommand{\qc}{\mathrm{QC}}
\newcommand{\normqc}{\mathrm{norm\mbox{-}QC}}
\newcommand{\rqc}{\mathsf{RQC\mbox{-}MC}}

\newcommand{\qcfo}{\mathrm{QCFO}}
\newcommand{\qcfofk}{\qcfo_{\forall}^k}
\newcommand{\qcfoek}{\qcfo_{\exists}^k}

\newcommand{\fo}{\mathrm{FO}}
\newcommand{\fok}{\mathrm{FO}^k}

\newcommand{\tup}[1]{\overline{#1}}

\newcommand{\nn}{\mathsf{nn}}
\newcommand{\bush}{\mathsf{bush}}
\newcommand{\width}{\mathsf{width}}

\newcommand{\un}{N^{\forall}}
\newcommand{\en}{N^{\exists}}

\newcommand{\ord}{\tup{u}}
\newcommand{\ordp}[1]{\tup{#1}}

\newcommand{\gc}{G^{-C}}

\newcommand{\ar}{\mathrm{ar}}
\newcommand{\free}{\mathsf{free}}
\newcommand{\vars}{\mathsf{vars}}

\newcommand{\f}{\mathcal{F}}

\newcommand{\pow}{\wp}

\newcommand{\N}{\mathbb{N}}

\newcommand{\param}[1]{\mathsf{param}\textup{-}#1}

\newcommand{\dom}{\mathsf{dom}}

\newcommand{\org}{\mathrm{org}^+}
\newcommand{\lay}{\mathrm{lay}^+}

\newcommand{\und}[1]{\underline{#1}}

\newcommand{\clo}{\mathsf{closure}}

\newcommand{\thick}{\mathsf{thick}}
\newcommand{\thickl}{\thick_l}
\newcommand{\localthickl}{\mathsf{local}\textup{-}\thickl}
\newcommand{\quantthickl}{\mathsf{quant}\textup{-}\thickl}

\newcommand{\lowerdeg}{\mathsf{lower}\textup{-}\mathsf{deg}}

\newcommand{\restrict}{\upharpoonright}

\renewcommand{\nu}[1]{\textup{{\small $\mathsf{nu}$}-{$ #1 $}}}

\newcommand{\case}[1]{\textup{{\small $\mathsf{case}$}-{$ #1 $}}}

\newcommand{\coclique}{\textup{{\small  $\mathsf{co}$}-{\small $\mathsf{CLIQUE}$}}}
\newcommand{\clique}{\textup{{\small $\mathsf{CLIQUE}$}}}

\newcommand{\caseclique}{\case{\clique}}
\newcommand{\casecoclique}{\case{\coclique}}

\newcommand{\fpt}{\textup{\small $\mathsf{FPT}$}}
\newcommand{\wone}{\textup{\small $\mathsf{W[1]}$}}
\newcommand{\cowone}{\textup{\small $\mathsf{co}$-$\mathsf{W[1]}$}}

\renewcommand{\S}{\mathcal{S}}
\newcommand{\G}{\mathcal{G}}

\section{Introduction} \label{section:introduction}

Model checking, the problem of deciding if a logical sentence holds on
a structure,
is a fundamental computational task that appears in many guises
throughout computer science.
In this article, we study \emph{first-order model checking},
by which we refer to
the case of this problem where one wishes to evaluate a first-order sentence
on a finite structure.
This case is of principal interest in database theory, where
first-order sentences form a basic, heavily studied class of \emph{database
  queries}, and where it is well-recognized that the problem of
evaluating
such a query on a database can be taken as a formulation of
first-order model checking.
Indeed, the investigation of model checking in first-order logic entails 
an examination of one of the simplest, most basic logics,
and it can be expected that understanding of the first-order case
should provide a well-founded basis for studying model checking
in other logics, such as those typically considered in verification
and database theory.
First-order model checking is well-known to be intractable in general:
it is PSPACE-complete.

As has been articulated~\cite{PapadimitriouYannakakis99-database,FlumGrohe06-parameterizedcomplexity},
the typical model-checking situation in the database and verification settings
is the evaluation of a relatively short sentence on a relatively large
structure.
Consequently, it has been argued that, in measuring the time
complexity
of model checking, one could reasonably allow a slow (non-polynomial-time)
preprocessing of the sentence, so long as the desired evaluation
can be performed in polynomial time following the preprocessing.
Relaxing polynomial-time computation 
to allow arbitrary preprocessing of a \emph{parameter} of a problem
instance
yields, 
in essence, the notion of \emph{fixed-parameter tractability}.
This notion of tractability is the base of \emph{parameterized
  complexity theory},
which provides a taxonomy for reasoning about and classifying
problems where each instance has an associated parameter.
We utilize this paradigm, and focus the discussion on this form of
tractability
(here, the sentence is the parameter).

A typical way to understand which types of sentences are well-behaved
and exhibit desirable, tractable behavior is to simply consider
model checking relative to a set $\Phi$ of sentences,
and to attempt to understand on which sets one has tractable model checking.
We restrict attention to sets of sentences having bounded arity.\footnote{
Note that in the case of unbounded arity, 
complexity may depend on the 
choice of representation of 
relations~\cite{ChenGrohe10-succinct}.}
Here, there have been successes in understanding 
which sets of sentences are tractable (and which are not) 
in fragments of first-order logic described by restricting
the connectives and quantifiers that may be used:
there are systematic classification results
for so-called conjunctive queries
(formed using the connectives and quantifiers 
in $\{ \wedge, \exists \}$)~\cite{GroheSchwentickSegoufin01-conjunctivequeries,Grohe07-otherside,ChenMueller13-fineclassification},
existential positive queries
($\{ \wedge, \vee, \exists \}$)~\cite{Chen14-existentialpositive}, 
and 
quantified conjunctive
queries
($\{ \wedge, \exists, \forall \}$)~\cite{ChenDalmau12-decomposingquantified,ChenMarx13-blocksorted}.
However, to the best of our knowledge, 
there has been no classification theorem for general first-order
logic,
without any restriction on the connectives and quantifiers.
In this article, we present the first such classification.

{\bf Our approach.}
In the fragments of first-order logic
where the only connective permitted is one of the binary connectives
($\{ \wedge, \vee \}$)---such as
those of conjunctive queries and 
quantified conjunctive queries---a heavily studied approach
to describing sets $\Phi$ of sentences is a \emph{graphical approach}.
In this graphical approach, one studies a sentence set $\Phi$
if it is \emph{graphical} in the following sense:
if one prenex sentence $\phi$ is contained in $\Phi$
and a second prenex sentence $\psi$ has the same prefix as $\phi$
and also has the same graph as $\phi$,
then $\psi$ is also in $\Phi$.
(By the graph of a prenex sentence $\phi$, we mean the graph whose
vertices are the variables of $\phi$ and where two vertices are 
adjacent if they occur together in an atomic formula.)
In the fragments where it was considered,
this approach 
of studying graphical sentence sets 
is not only a natural way to coarsen 
the project of classifying all sets $\Phi$ of sentences,
but in fact, can be used cleanly as a key module in 
obtaining general classifications of sentence sets:
such general classifications have recently been proved by using
the respective graphical classifications as 
black boxes~\cite{ChenMueller13-fineclassification,ChenMarx13-blocksorted}.

In this article, we adapt this graphical approach to the full
first-order setting.
To explain how this is done,
consider that a graphical set $\Phi$ of sentences satisfies
certain syntactic closure properties.
For instance, if one takes a sentence from such a graphical set $\Phi$
and
replaces the relation symbol of an atomic formula,
the resulting sentence will have the same graph,
and will hence continue to be in $\Phi$;
we refer to this property of $\Phi$ as \emph{replacement closure}.
As another example,
if one rewrites a sentence in $\Phi$ by invoking associativity 
or commutativity of the connective $\wedge$,
the resulting sentence will likewise still have the same graph
and will hence be contained in $\Phi$ also.
Inspired by these observations, we define 
a sentence set $\Phi$ to be \emph{graph-like}
if it is replacement closed and
also closed under certain well-known syntactic transformations,
such as associativity and commutativity of the binary connectives
(see Section~\ref{sect:preliminaries} for the full definition).

Our principal result is the complete tractability characterization
of graph-like sentence sets
(see Theorem~\ref{thm:main} and the corollaries that follow).
In particular,
we introduce a measure on first-order formulas which we call
\emph{thickness},
and show that a set of graph-like sentences is tractable
if and only if it has bounded thickness
(under standard complexity-theoretic assumptions).
In studying unrestricted first-order logic,
we believe that our building on the syntactic,
graphical approach---which has an established, fruitful
tradition---will
facilitate the formulation and obtention of future, more general results.

As evidence of our result's generality 
and of its faithfulness to the graphical approach,
we 
note
that the graphical classification
of quantified conjunctive queries~\cite{ChenDalmau12-decomposingquantified}
can be readily derived from our main classification result
(this is discussed in Section~\ref{sect:discussion});
it follows readily that the dual graphical classification
of quantified disjunctive queries
(also previously derived~\cite{ChenDalmau12-decomposingquantified})
can be dually derived from our main classification result.
We therefore give a single classification theorem
that naturally unifies together these two previous classifications.
Indeed, we believe that the technology that we introduce to derive
our classification yields a cleaner, deeper and more general
understanding of these previous classifications.
Observe that, for each of those classifications,
since only one binary connective is present, the two quantifiers
behave asymmetrically;
this is in contrast to the present situation, where in building
formulas,
wherever a formula may be constructed, its dual may be as well.

{\bf Parameterized complexity.}
An increasing literature investigates the following general situation:
\begin{center}
Given a parameterized problem $P$ \\whose instances
consist of two parts, where the first part is the parameter, and a set $S$,
\\define $P \llbracket S \rrbracket$ to be the restricted version of $P$\\
where $(x, y)$ is admitted as an instance iff $x \in S$.

\vspace{4pt}

Then, attempt to classify and understand, over all sets $S$,\\
the complexity of the problem $P \llbracket S \rrbracket$.
\end{center}
Examples of classifications and studies that can be cast in this situation
include~\cite{GroheSchwentickSegoufin01-conjunctivequeries,Grohe07-otherside,ChenThurleyWeyer08-inducedsubgraph,ChenGrohe10-succinct,ChenDalmau12-decomposingquantified,DurandMengel13-structuralcounting,ChenMueller13-fineclassification,JonssonLagerkvistNordh13-blowingholes}.
It is our view that this literature suffers from the defect
that there is no complexity-theoretic framework for discussing the 
families of problems obtained thusly.
As a consequence, different authors and different articles
used divergent language and notions
to present hardness results on and reductions between such problems,
and
applied
different computability assumptions on the sets $S$ 
considered.
We attempt to make a foundational contribution 
and to ameliorate this state of affairs
by presenting a complexity-theoretic framework
for handling and classifying problems of the described form.
In particular, we introduce notions such as reductions and
complexity classes 
for problems of the above type, which we formalize as \emph{case problems}
(Section~\ref{sect:case}).
Although we do not carry out this exercise here,
we believe that most of the results in the mentioned literature
can be shown to be naturally and transparently
expressible within our framework.

In order to derive our classification, we present
(within our complexity-theoretic framework) a new notion of reduction
which we call \emph{accordion reduction} and
which is crucial for the proof of our hardness result.
(See Sections~\ref{sect:accordion} and~\ref{sect:hardness} for further discussion.)
We believe that this notion of reduction may play a basic role
in future classification projects of the form undertaken here.

Let us emphasize that, while the establishment of our
main classification theorem makes use of the complexity-theoretic
framework and accompanying machinery that was just discussed
(and is presented in Sections~\ref{sect:case} and~\ref{sect:accordion}),
this framework and machinery is fully generic
in that it does not make any reference to and is not specialized to
the model checking problem.
We believe and hope that the future will find this framework
to be a suitable basis for presenting, developing and discussing
complexity classification results.

\section{Preliminaries}
\label{sect:preliminaries}

When $g: A \to B$ and $h: B \to C$ are mappings, we will 
typically use $h(g)$ to denote their composition.
When $f$ is a partial mapping, we use $\dom(f)$ to denote its domain,
and we use $f \restrict S$ to denote its restriction to the set $S$.
We will use $\pi_i$ to denote the $i$th projection,
that is, the mapping that, given a tuple, returns the value in the tuple's
$i$th coordinate.
For a natural number $k$, we use $\und{k}$ to denote the set $\{ 1, \ldots, k \}$.

\subsection{First-order logic}
We use the syntax and semantics of first-order logic
as given by a standard treatment of the subject.
In this article, we restrict to relational first-order logic,
so the only symbols in signatures are relation symbols.
We use letters such as $\rela$, $\relb$ to denote structures
and $A$, $B$ to denote their respective universes.
The reader may assume for concreteness that relations of structures
are represented using lists of tuples,
although in general, we will deal with the setting of bounded arity,
and natural representations of relations will be 
(for the complexity questions at hand) equivalent to this one.
We assume that equality is not built-in to first-order logic,
so a \emph{formula} is created from \emph{atoms},
the usual connectives ($\neg$, $\wedge$, $\vee$)
and quantification ($\exists$, $\forall$);
by an \emph{atom}, we mean a formula
$R(v_1, \ldots, v_k)$ where a relation symbol 
is applied to a tuple of variables (of the arity of the symbol).
A formula is \emph{positive} if it does not contain negation ($\neg$).
We use $\free(\phi)$ to denote the set of free variables
of a formula $\phi$.
The \emph{width} of a formula $\phi$,
denoted by $\width(\phi)$, is 
defined as the maximum of $|\free(\psi)|$
over all subformulas $\psi$ of $\phi$.
The \emph{arity} of a formula is the maximum arity over all relation
symbols
that occur in the formula.

We use 
$\phi_1 \wedge \cdots \wedge \phi_n$ as notation 
for $(\cdots ((\phi_1 \wedge \phi_2) \wedge \phi_3) \cdots)$,
and 
$\phi_1 \vee \cdots \vee \phi_n$ 
is defined dually.
We refer to a formula of the shape
$\phi_1 \wedge \cdots \wedge \phi_n$ 
as a \emph{conjunction},
and respectively refer to a formula of the shape
$\phi_1 \vee \cdots \vee \phi_n$ 
as a \emph{disjunction}.
Let $\Phi$ be a set of formulas.
A \emph{positive combination} of formulas from $\Phi$
is a formula in the closure of $\Phi$ under conjunction and
disjunction.
A \emph{CNF} of formulas from $\Phi$ is a 
conjunction of disjunctions of formulas from $\Phi$, and
a \emph{DNF} of formulas from $\Phi$ is a
disjunction of conjunctions of formulas from $\Phi$.

We will use the following terminology which is particular to this article.
A subformula $\psi$ of a formula $\phi$ is
a \emph{positively combined subformula} 
if, in viewing $\phi$ as a tree, all nodes 
on the unique path
from the root of 
$\phi$
to the parent of the root of $\psi$ (inclusive)
are conjunctions or disjunctions.
We say that a formula $\phi$ is \emph{variable-loose} if
no variable is quantified twice, and no variable is both quantified
and a free variable of $\phi$.
We say that a formula is \emph{symbol-loose} if
no relation symbol appears more than once in the formula.
We say that a formula is \emph{loose} 
if it is both variable-loose and symbol-loose.

We now present a number of syntactic transformations;
that each preserves logical equivalence is well-known.\footnote{
Note that in the transformation $(\gamma)$, we permit that $\phi$ 
is not present.}

\begin{itemize}

\item[($\alpha$)] Associativity and commutativity of $\wedge$ and
  $\vee$

\item [($\beta$)] 
$\exists x (\bigvee_{i = 1}^n \phi_i) \equiv \bigvee_{i=1}^n (\exists
x \phi_i)$,

$\forall y (\bigwedge_{i = 1}^n \phi_i) \equiv \bigwedge_{i=1}^n (\forall y \phi_i)$

\item [($\gamma$)] 
$\exists x (\phi \wedge \psi) \equiv (\exists x \phi) \wedge \psi
\textup{ if $x \notin \free(\psi)$}$,

$\forall y (\phi \vee \psi) \equiv (\forall y \phi) \vee \psi
\textup{ if $y \notin \free(\psi)$}$

\item [($\delta$)] (Distributivity for $\wedge$ and $\vee$)

$\phi \wedge (\psi \vee \psi') \equiv 
(\phi \wedge \psi) \vee (\phi \wedge \psi')$, 

$\phi \vee (\psi \wedge \psi') \equiv 
(\phi \vee \psi) \wedge (\phi \vee \psi')$

\item [($\epsilon$)] (DeMorgan's laws)

$\neg \exists v \phi \equiv \forall v \neg \phi$,
$\neg \forall v \phi \equiv \exists v \neg \phi$

$\neg (\phi \wedge \psi) \equiv \neg \phi \vee \neg \psi$,
$\neg (\phi \vee \psi) \equiv \neg \phi \wedge \neg \psi$

\end{itemize}
We say that a set $\Phi$ of formulas is \emph{syntactically closed} if,
for each $\phi \in \Phi$, when a formula $\phi'$ can be obtained from
$\phi$ by applying one of the syntactic transformations
$(\alpha)$, $(\beta)$, $(\gamma)$, $(\delta)$, $(\epsilon)$
 to a subformula of
$\phi$,
it holds that $\phi' \in \Phi$.
The \emph{syntactic closure} of a formula $\phi$ is the 
intersection of all syntactically closed sets that contain $\phi$.

Let us say that a formula $\phi'$ on signature $\sigma'$
is obtainable from a formula $\phi$ on signature $\sigma$
by \emph{replacement} if $\phi'$ can be obtained from $\phi$ by 
replacing instances of relation symbols in $\phi$ 
with instances of relation symbols from $\sigma'$
(without making any other changes to $\phi$).

\begin{example}
Let $\phi$ be the formula $\forall y \exists x \exists x' (E(y, x) \wedge E(x, x'))$.
Let $\tau$ be the signature $\{ E, F \}$ where $E$ and $F$ are
relation symbols of binary arity.
Each of the four formulas $\phi$,
$\forall y \exists x \exists x' (F(y, x) \wedge E(x, x'))$,
$\forall y \exists x \exists x' (E(y, x) \wedge F(x, x'))$, and
$\forall y \exists x \exists x' (F(y, x) \wedge F(x, x'))$
is obtainable from $\phi$ by replacement;
moreover, these are the only four formulas over signature $\tau$
that are obtainable from $\phi$ by replacement.
\end{example}

Let us say that a set of formulas $\Phi$ is \emph{replacement closed}
if, for each $\phi \in \Phi$, when $\phi'$ is obtainable from $\phi$
by replacement, it holds that $\phi' \in \Phi$.

\begin{definition}
A set of formulas $\Phi$ is \emph{graph-like} if it is syntactically
closed
and replacement closed.
\end{definition}

\subsection{Graphs and hypergraphs}
When $S$ is a set, we use $K(S)$ to denote the set containing
all size $2$ subsets of $S$, that is,
$K(S) = \{ \{ s, s' \} ~|~ s, s' \in S, s \neq s' \}$.
For us, a \emph{graph} is a pair $(V, E)$ where $V$ is a set
and $E \subseteq K(V)$.

Here, a \emph{hypergraph} $H$ is a pair $(V(H), E(H))$
consisting of a vertex set $V(H)$ and an edge set $E(H)$
which is a subset of the power set $\pow(V(H))$.
We will sometimes specify a hypergraph just by specifying the edge
set $E$,
in which case the vertex set is understood to be $\bigcup_{e \in E} e$.
We associate a hypergraph $(V(H), E(H))$ with the graph
$(V(H), \bigcup_{e \in E(H)} K(e))$, and thereby refer to (for
example)
the treewidth of or an elimination ordering of a hypergraph.

An \emph{elimination ordering} of a graph $(V, E)$
is a pair $$((v_1, \ldots, v_n), E')$$ that
consists of a superset $E'$ of $E$ and an ordering
$v_1, \ldots, v_n$ of the elements of $V$ such that
the following property holds: for each vertex $v_k$,
any two distinct lower neighbors $v, v'$ of $v_k$ are
adjacent in $E'$, that is, $\{ v, v' \} \in E'$; 
here, a \emph{lower neighbor} of a vertex $v_k$
is a vertex $v_i$ such that $i < k$ and $\{ v_i, v_k \} \in E'$.
Relative to an elimination ordering $e$, 
we define the \emph{lower degree} of a vertex $v$,
denoted by $\lowerdeg(e, v)$,
 to be the
number of lower neighbors that it has;
we define $\lowerdeg(e)$ to be the maximum of
$\lowerdeg(e, v)$ over all vertices $v$.
We assume basic familiarity with the theory of treewidth~\cite{Bodlaender98-arboretum}.
The following is a key property of treewidth that we will utilize;
here, we use $\tw(H)$ to denote the treewidth of $H$.
\begin{prop}
\label{prop:tw-gives-elimination-ordering}
For each $k \geq 2$, there exists a polynomial-time algorithm that,
given as input a hypergraph $H$ with a distinguished edge $f$,
will return the following whenever $\tw(H) < k$: 
an elimination ordering $e = ((v_1, \ldots, v_m), E)$
with $\lowerdeg(e) = \tw(H)$ and where
$\{ v_1, \ldots, v_{|f|} \} = f$.
\end{prop}

The treewidth of (for example) a set of graphs $\fancyg$
is the set $\{ \tw(G) ~|~ G \in \fancyg \}$;
it is said to be \emph{unbounded} if this set is infinite,
and \emph{bounded} otherwise.
We employ similar terminology, in general, when dealing with
a complexity measure defined on a class of objects.

\section{Parameterized complexity}

In this section, we specify the framework of parameterized complexity
to be used in this article.

Throughout, we use $\Sigma$ to denote an alphabet over which languages
are defined.
As is standard, we will sometimes view elements of $\Sigma^* \times
\Sigma^*$
as elements of $\Sigma^*$.
A \emph{parameterization} is a mapping from $\Sigma^*$ to $\Sigma^*$.
A \emph{parameterized problem} is a pair $(Q, \kappa)$ 
consisting of a language $Q \subseteq \Sigma^*$ and a parameterization
$\kappa: \Sigma^* \to \Sigma^*$.
A \emph{parameterized class} is a set of parameterized problems.

\begin{assumption}
\label{assumption:non-trivial-languages}
We assume that each parameterized problem $(Q, \kappa)$
has a non-trivial language $Q$, that is, that neither $Q = \Sigma^*$
nor $Q = \emptyset$.
\end{assumption}

\begin{remark}
Let us remark on a difference between our setup and that of other
treatments.
Elsewhere, a \emph{parameterization} is often defined
to be a mapping from $\Sigma^*$ to $\N$, and in the context of
query evaluation, the parameterization studied is typically
the size of the query.  
In contrast, this article takes the parameterization to be the query itself.
Since there are finitely many queries of any fixed size, 
model checking on a set of queries will be fixed-parameter tractable
(that is, in the class $\fpt$, defined below)
under one of these parameterizations if and only if it is under the other.
However, we find that---as concerns the theory in this
article---taking 
the query itself to be the parameter
allows for a significantly cleaner presentation.
One example reason is that
the reductions we present will generally be ``slice-to-slice'',
that is, they will send all instances with the same query to instances
that share another query.
Indeed,
to understand the complexity of a set of queries, we will 
apply a closure operator to pass to a larger set of queries
having the same complexity, using the notion of 
\emph{accordion reduction}
(see Sections~\ref{sect:accordion}
and~\ref{sect:hardness});
we believe that the theory
justifying this passage is most cleanly expressed under the used
parameterization.

Let us also remark that we do not put in effect any background
assumption on the
computability/complexity
of parameterizations; this is because our theory does not 
require such an assumption.
%
$\footbox$
\end{remark}

We now define what it means for a partial mapping $r$
to be \emph{FPT-computable}; we actually first define
a non-uniform version of this notion.
This definition coincides with typical definitions in the case 
that $r$ is a total mapping; 
in the case that $r$ is a partial mapping, we use a ``promise''
convention,
that is, we do not impose any mandate on the behavior
of the respective algorithm in the case that the input $x$
is not in the domain of $r$.

\begin{definition}
\label{def:fpt-computable}
Let $\kappa: \Sigma^* \to \Sigma^*$ be a parameterization.

A partial mapping $r: \Sigma^* \to \Sigma^*$ is 
\emph{nu-FPT-computable
with respect to $\kappa$}
if there exist a function $f: \Sigma^* \to \N$ and a polynomial $p: \N \to \N$
such that
for each $k \in \Sigma^*$, 
there exists an algorithm $A_k$ satisfying the following condition:
on each string $x \in \dom(r)$ such that $\kappa(x) = k$,
the algorithm $A_k$ computes $r(x)$ within time
$f(\kappa(x))p(|x|)$.

A partial mapping $r: \Sigma^* \to \Sigma^*$ is 
\emph{FPT-computable
with respect to $\kappa$}
if, in the just-given definition, 
the function $f$ can be chosen to be computable
and
 there exists a single algorithm $A$ that can play
the role of each algorithm $A_k$;
formally,
if there exist a computable function $f: \Sigma^* \to \N$, 
a polynomial $p: \N \to \N$,
and an algorithm $A$ such that 
for each $k \in \Sigma^*$, the above condition is satisfied when $A_k$ is
set equal to $A$.
$\footbox$
\end{definition}

\begin{definition}
We define $\fpt$ to be the class that contains a parameterized problem
$(Q, \kappa)$
if and only if
the characteristic function of $Q$ is FPT-computable
with respect to $\kappa$.
\end{definition}

\newcommand{\powfin}{\pow_{\mathsf{fin}}}

We now introduce the notion of a reduction between parameterized
problems.  When $A$ is a set, we use $\powfin(A)$ to denote
the set containing all finite subsets of $A$.

\begin{definition}
Let $(Q, \kappa)$ and $(Q', \kappa')$
be parameterized problems.
A \emph{FPT-reduction} (respectively, \emph{nu-FPT-reduction}) 
from $(Q, \kappa)$ to $(Q', \kappa')$
is a total mapping $g: \Sigma^* \to \Sigma^*$ that 
is FPT-computable 
(respectively, nu-FPT-computable) 
with respect to $\kappa$ 
and
a computable (respectively, not necessarily computable)
mapping  $h: \Sigma^* \to \powfin(\Sigma^*)$
such that:
\begin{itemize}
\item[(1)] for each $x \in \Sigma^*$, it holds that
$x \in Q$ if and only if $g(x) \in Q'$; and
\item[(2)] for each $x \in \Sigma^*$,
it holds that $\kappa'(g(x)) \in h(\kappa(x))$.
\end{itemize}
\end{definition}

\begin{assumption}
\label{assumption:closure-under-reduction}
We assume that each parameterized class $C$
is closed under FPT-reductions, that is,
if $(Q', \kappa')$ is in $C$ and 
$(Q, \kappa)$ FPT-reduces to $(Q', \kappa')$,
then $(Q, \kappa)$ is in $C$.
\end{assumption}

\begin{prop}
\label{prop:fpt-reductions-compose}
Let $(Q, \kappa)$, $(Q', \kappa')$, and $(Q'', \kappa'')$
be parameterized problems.
\begin{itemize}

\item 
If 
$g$ is a FPT-reduction from $(Q, \kappa)$ to $(Q', \kappa')$
and
$h$ is a FPT-reduction from $(Q', \kappa')$ to $(Q'', \kappa'')$,
then their composition $h(g)$ 
is a FPT-reduction from $(Q, \kappa)$ to $(Q'', \kappa'')$.

\item Similarly,
if 
$g$ is a nu-FPT-reduction from $(Q, \kappa)$ to $(Q', \kappa')$
and
$h$ is a nu-FPT-reduction from $(Q', \kappa')$ to $(Q'', \kappa'')$,
then their composition $h(g)$ 
is a nu-FPT-reduction from $(Q, \kappa)$ to $(Q'', \kappa'')$.

\end{itemize}
Also, each FPT-reduction from $(Q, \kappa)$ to $(Q', \kappa')$
is an nu-FPT-reduction from $(Q, \kappa)$ to $(Q', \kappa')$.
\end{prop}

We may now define, for each parameterized class $C$,
a non-uniform version of the class, denoted by $\nu{C}$.

\begin{definition} (non-uniform classes)
When $C$ is a parameterized class, we define $\nu{C}$
to be the set that contains each parameterized problem
that has a nu-FPT-reduction to a problem in $C$.
\end{definition}

\begin{remark}
It is straightforwardly verified that a parameterized problem
$(Q, \kappa)$ 
is in the class $\nu{\fpt}$ (under the above definitions) 
if and only if the characteristic function of $Q$ is
nu-FPT-computable with respect to $\kappa$.
\end{remark}

We next present two notions of hardness for parameterized classes.

\begin{definition} (hardness)
Let $C$ be a parameterized class.
We say that
a problem $(Q, \kappa)$ is $C$-hard if every problem in 
$C$ has a FPT-reduction to $(Q, \kappa)$.
We say that 
a problem $(Q, \kappa)$ is non-uniformly $C$-hard if every problem in
$C$ has a nu-FPT-reduction to $(Q, \kappa)$.
\end{definition}

We end this section by observing some basic closure properties of
what we call \emph{degree-bounded functions},
which, roughly speaking, are the functions which can serve as
the running time of algorithms in the definition of 
\emph{FPT-computable} (Definition~\ref{def:fpt-computable}).
Let $\kappa: \Sigma^* \to \Sigma^*$ be a parameterization.
A partial function $T: \Sigma^* \to \N$ is
\emph{degree-bounded} with respect to $\kappa$
if there exist a computable function $f: \Sigma^* \to \N$
and a polynomial $p: \N \to \N$ such that,
for each $x \in \dom(T)$,
it holds that $T(x) \leq f(\kappa(x)) p(|x|)$.
We will make use of the observation that
a partial mapping $h: \Sigma^* \to \Sigma^*$ is FPT-computable
with respect to $\kappa$ 
if and only if
there is a degree-bounded function $T$
(with $\dom(T) \supseteq \dom(h)$)
and an algorithm
that, for all $x \in \dom(h)$,
computes $h(x)$ within time $T(x)$.

The following proposition will be of use in establishing that
functions are degree-bounded.

\begin{prop}
\label{prop:db-closure}
Let $\kappa$ be a parameterization,
and let $T_1, \ldots, T_m: \Sigma^* \to \N$ be
partial functions sharing the same domain.

\begin{enumerate}

\item 
\label{closure:sum}
If each of $T_1, \ldots, T_m$
is degree-bounded with respect to $\kappa$, then
$T_1 + \cdots + T_m$ is as well.

\item 
\label{closure:product}
If each of $T_1, \ldots, T_m$
is degree-bounded with respect to $\kappa$, then
the product $T_1 \cdots T_m$ is as well.

\item 
\label{closure:polynomial}
Let $q: \N \to \N$ be a polynomial; 
if a partial function $T: \Sigma^* \to \N$
is degree-bounded with respect to $\kappa$, then
$q(T)$ is as well.


\end{enumerate}
\end{prop}


\begin{proof}
For~(\ref{closure:sum}), one can use the fact that
$f_1(k) p_1(n) + \cdots + f_m(k) p_m(n)$ is bounded above by
$(f_1(k) + \cdots + f_m(k))(p_1(n) + \cdots + p_m(n))$.

For~(\ref{closure:product}), it suffices to observe that the product
$(f_1(k) p_1(n)) \cdots (f_m(k) p_m(n))$ can be grouped as
$(f_1(k) \cdots f_m(k))(p_1(n) \cdots p_m(n))$.

For~(\ref{closure:polynomial}), it suffices to observe that
$q(f(k)p(n))$ is bounded above by $q(f(k)) q(p(n))$;
note that $q(p)$ is the composition of two polynomials, and hence
itself a polynomial.
%
\end{proof}

\section{Case complexity}
\label{sect:case}

A number of previous works focus on a decision problem $Q$ where
each instance consists of two parts, 
and obtain restricted versions
of the problem
by taking sets $S \subseteq \Sigma^*$
and, for each such set, considering the restricted version
where one allows only instances where (say) the first of the parts
falls into $S$.
This is precisely the type of restriction that we will consider here.
We are interested in first-order model checking, which 
we view as the problem of deciding, given a first-order sentence and a
structure,
whether or not the sentence holds on the structure;
our particular interest is to study restricted versions of this
problem 
where the allowed sentences come from a set $S$.

It has been useful 
(see for instance the articles~\cite{ChenDalmau12-decomposingquantified,ChenMueller13-fineclassification})
and is useful in the present article to 
present reductions between such restricted versions of problems.
In order to facilitate our doing this, we present a framework
wherein we formalize this type of restricted version of problem
as a \emph{case problem}, and then present a notion of reduction
for comparing case problems.
We believe that our notion of reduction, called \emph{slice
  reduction},
faithfully abstracts out precisely the key useful properties 
that are typically present in such reductions in the literature.
Note that, in the existing literature, different articles imposed
different computability assumptions on the sets $S$ considered
(assumptions used include that of computable enumerability, 
of computability, and of no computability assumption).
One feature of our framework is that such reductions can be carried
out and discussed independently of whether or not any such 
computability assumption is placed on the sets $S$;
a general theorem (Theorem~\ref{thm:slice-reduction-gives-fpt-reduction}) 
allows one to derive
normal parameterized reductions from slice reductions,
where the exact computability of the reduction derivable
depends on the computability assumption placed on the sets $S$.

We now introduce our framework.
Suppose that $Q \subseteq \Sigma^* \times \Sigma^*$ 
is a language of pairs; for a set $T \subseteq \Sigma^*$,
we use $Q_T$ to denote the language
$Q \cap (T \times \Sigma^*)$
and for a single string $t \in \Sigma^*$,
we use $Q_t$ to denote the language
$Q \cap (\{ t \} \times \Sigma^*)$.

\begin{definition}
A \emph{case problem} consists of a language of pairs
$Q \subseteq \Sigma^* \times \Sigma^*$
and a subset $S \subseteq \Sigma^*$,
and is denoted $Q[S]$.
When $Q[S]$ is a case problem, 
we use $\param{Q[S]}$ to denote the parameterized problem
$(Q_S, \pi_1)$.
\end{definition}
Ultimately, our purpose in discussing a case problem $Q[S]$
is to understand the complexity of the associated parameterized problem
$\param{Q[S]}$.  As mentioned, formalizing the notion of a case
problem
allows us to cleanly present reductions between such problems.

\begin{remark}
Let $(Q, \kappa)$ be a parameterized problem.
Under the assumption that $\kappa$ is FPT-computable
with respect to itself,
the parameterized problem $(Q, \kappa)$
is straightforwardly verified to be equivalent, under FPT-reduction,
to the parameterized problem $(Q', \pi_1)$
where $Q' = \{ (\kappa(x), x) ~|~ x \in \Sigma^* \}$.
Hence, any such given parameterized problem $(Q, \kappa)$ 
may be canonically associated
to the case problem $Q'[\Sigma^*]$, as one has
$\param{Q'[\Sigma^*]} = (Q', \pi_1)$.
\end{remark}

We define $\caseclique$ to be the case problem 
$Q[\Sigma^*]$ where $Q$ contains a pair $(k, G)$
if and only if $G$ is a graph that contains a clique of size $k$
(we assume that both $G$ and $k$ are encoded as strings over
$\Sigma$);
we define $\casecoclique$ to be the problem $\overline{Q}[\Sigma^*]$.

We now present the notion of \emph{slice reduction},
which allows us to compare case problems.

\begin{definition}
A case problem $Q[S]$ \emph{slice reduces}
to a second case problem $Q'[S']$
if there exist:

\begin{itemize}

\item a computably enumerable language 
$U \subseteq \Sigma^* \times \Sigma^*$ 
and 

\item a partial function $r: \Sigma^* \times \Sigma^* \times \Sigma^* \to \Sigma^*$
that 
has $\dom(r) = U \times \Sigma^*$ and 
is FPT-computable with respect to the parameterization $(\pi_1, \pi_2)$

\end{itemize}
such that
the following conditions hold:
\begin{itemize}

\item \emph{(coverage)}
for each $s \in S$, 
there exists $s' \in S'$ such that $(s, s') \in U$, and

\item \emph{(correctness)} 
for each $(t, t') \in U$, it holds
(for each $y \in \Sigma^*$) that
$$(t, y) \in Q \Leftrightarrow (t', r(t, t', y)) \in Q'.$$

\end{itemize}
We call the pair $(U, r)$ a \emph{slice reduction} 
from $Q[S]$ to $Q'[S']$.
\end{definition}
%
%
%
In this definition, we understand the parameterization 
$(\pi_1, \pi_2)$ to be the mapping that, for all $(s, s') \in U$ and 
$y \in \Sigma^*$,
returns the pair $(s, s')$ given the triple $(s, s', y)$.

\begin{theorem}
\label{thm:slice-reduction-transitive}
(Transitivity of slice reducibility)
Suppose 
that $Q_1[S_1]$ slice reduces to $Q_2[S_2]$ 
and
that $Q_2[S_2]$ slice reduces to $Q_3[S_3]$. 
Then $Q_1[S_1]$ slice reduces to $Q_3[S_3]$.
\end{theorem}

The following theorem allows one to derive an FPT-reduction
or an nu-FPT-reduction from a slice reduction.

\begin{theorem}
\label{thm:slice-reduction-gives-fpt-reduction}
Suppose that a case problem $Q[S]$ slice reduces to 
another case problem $Q'[S']$.
Then, it holds that 
$\param{Q[S]}$ nu-FPT-reduces to $\param{Q'[S']}$;
if in addition $S$ and $S'$ are computable,
then 
$\param{Q[S]}$ FPT-reduces to $\param{Q'[S']}$.
\end{theorem}


\begin{proof}
Let $(U, r)$ be the slice reduction.
First, consider the case where both $S$ and $S'$ are computable.
Since $S$ and $S'$ 
are both computable and $U$ is computably enumerable,
there exists a computable function
 $f: \Sigma^* \to \Sigma^*$
such that, for all $s \in S$, it holds that $(s, f(s)) \in U$.
We define the reduction $g$ so that
$g(s, x)$ is equal to $(f(s), r(s, f(s), x))$
for all $(s, x) \in S \times \Sigma^*$,
and is otherwise defined to be a fixed string outside of $Q'$
(such a string exists by
Assumption~\ref{assumption:non-trivial-languages}).
The mapping $g$ has a natural algorithm,
namely, given $(s, x)$, check if $s \in S$ (via an algorithm for $S$);
if $s \notin S$, return a fixed string outside of $Q'$,
otherwise, compute the value of $g$ using an algorithm for $f$
and an algorithm witnessing FPT-computability of $r$.
Suppose that $s \in S$;
the value $f(s)$, viewed as a function of $(s, x)$,
is FPT-computable with respect to $\pi_1$;
since $r$ is FPT-computable with respect to $(\pi_1, \pi_2)$,
the value $r(s, f(s), x)$, viewed as a function of $(s, x)$,
is also FPT-computable with respect to $\pi_1$.
Thus, the mapping $g$ is FPT-computable with respect to $\pi_1$.

In the case that no computability assumptions are placed on $S$ and
$S'$,
there exists a function
 $f: \Sigma^* \to \Sigma^*$
such that, for all $s \in S$, it holds that $(s, f(s)) \in U$.
The reduction $g$ defined as above is readily verified to 
be nu-FPT-computable with respect to $\pi_1$,
via an ensemble of algorithms where $A_s$ contains as hard-coded
information whether or not $s \in S$ and (if so) the value of $f(s)$.
\end{proof}


For each parameterized class $C$, we define
$\case{C}$ to be
the set of case problems 
that contains
a case problem $Q[S]$ if and only if
there exists a case problem $Q'[S']$ such that
\begin{itemize}

\item $\param{Q'[S']}$ is in $C$,

\item $Q[S]$ slice reduces to $Q'[S']$, and

\item $S'$ is computable.

\end{itemize}





\begin{prop}
\label{prop:in-casec}
Suppose that 
$C$ is a parameterized class, and that
$Q[S]$ is a case problem in $\case{C}$.
Then, the problem $\param{Q[S]}$ is in $\nu{C}$;
if it is assumed additionally that $S$ is computable,
then the problem $\param{Q[S]}$ is in $C$.
\end{prop}


\begin{proof}
There exists a case problem $Q'[S']$ satisfying the conditions
given in the definition of $\case{C}$.
Since $Q[S]$ slice reduces to $Q'[S']$, 
by Theorem~\ref{thm:slice-reduction-gives-fpt-reduction},
it holds 
that $\param{Q[S]}$ nu-FPT-reduces to $\param{Q'[S']}$.
Since $\param{Q'[S']}$ is in $C$, it follows that
$\param{Q[S]}$ is in $\nu{C}$.
If in addition it is assumed that $S$ is computable,
by Theorem~\ref{thm:slice-reduction-gives-fpt-reduction},
we obtain 
that $\param{Q[S]}$ FPT-reduces to $\param{Q'[S']}$,
and hence that $\param{Q[S]}$ is in $C$
(by appeal to Assumption~\ref{assumption:closure-under-reduction}).
\end{proof}

Let $Q[S]$ be a case problem and let $C$ be a parameterized class.
We say that $Q[S]$ is $\case{C}$-hard if there exists a case problem
$Q^-[S^-]$ such that
\begin{itemize}

\item $\param{Q^-[S^-]}$ is $C$-hard,

\item $Q^-[S^-]$ slice reduces to $Q[S]$, and

\item $S^-$ is computable.

\end{itemize}

\begin{prop}
\label{prop:casec-hard}
Suppose that $C$ is a parameterized class, and that
$Q[S]$ is a case problem that is $\case{C}$-hard.
Then, the problem $\param{Q[S]}$ is non-uniformly $C$-hard;
if it is assumed additionally that $S$ is computable,
then the problem $\param{Q[S]}$ is $C$-hard.
\end{prop}

\begin{proof}
There exists a case problem $Q^-[S^-]$ satisfying
the conditions given in the definition of $\case{C}$-hard.
Since $Q^-[S^-]$ slice reduces to $Q[S]$,
by Theorem~\ref{thm:slice-reduction-gives-fpt-reduction},
it holds that
$\param{Q^-[S^-]}$ nu-FPT-reduces to 
$\param{Q[S]}$.  
By the $C$-hardness of $\param{Q^-[S^-]}$,
each problem in $C$ FPT-reduces to $\param{Q^-[S^-]}$,
and hence, by Proposition~\ref{prop:fpt-reductions-compose},
each problem in $C$ nu-FPT-reduces to 
$\param{Q[S]}$.  
If in addition it is assumed that $S$ is computable,
by Theorem~\ref{thm:slice-reduction-gives-fpt-reduction},
the problem $\param{Q^-[S^-]}$ FPT-reduces to 
the problem $\param{Q[S]}$;
from the $C$-hardness of $\param{Q^-[S^-]}$
and from Proposition~\ref{prop:fpt-reductions-compose},
we obtain that $\param{Q[S]}$ is $C$-hard.
\end{proof}

\section{Thickness}
\label{sect:thickness}

In this section, we define a measure 
of first-order formulas that we call \emph{thickness},
which we will show is the crucial measure that determines
whether or not a graph-like set of sentences is tractable.
From a high-level viewpoint, the measure is defined in the following
way.
We first 
define a notion of \emph{organized formula} and
show that each formula is logically equivalent to a positive
combination of organized formulas;
we then define a notion of \emph{layered formula}
and show that each organized formula is logically equivalent
to a layered formula.
We then,
for each layered formula $\phi$, 
define its thickness  (denoted by $\thickl(\phi)$),
and then naturally extend this definition to positive combinations
of layered formulas, and hence to all formulas.
A key property of thickness, which we prove
(Theorem~\ref{thm:thick-many-variables}),
is that there exists an algorithm that, given a formula $\phi$,
outputs an equivalent formula that uses at most $\thick(\phi)$ many
variables.

\begin{definition}
We define the set of organized formulas inductively, as follows.
\begin{itemize}

\item Each atom and each negated atom is an organized formula.

\item If each of $\phi_1, \ldots, \phi_n$ is an organized formula and
$v \in \free(\phi_1) \cap \cdots \cap \free(\phi_n)$, then
$\exists v (\bigwedge_{i=1}^n \phi_i)$
and
$\forall v (\bigvee_{i=1}^n \phi_i)$
are organized formulas.

\end{itemize}
\end{definition}

\begin{theorem}
\label{thm:formula-to-pos-comb-organized}
There exists an algorithm $\org$ that, given a formula $\phi$ as input, outputs
a positive combination $\org(\phi)$ of organized formulas that is logically
equivalent to $\phi$ and that is in the syntactic closure of $\phi$.
\end{theorem}

The algorithm of this theorem is defined recursively with respect to formula structure.


\begin{proof}
To give the algorithm,
we first recursively define a procedure that,
given a formula $\phi$ where negations appear only in front of atoms, 
outputs both a CNF of organized formulas
and a DNF of organized formulas, each of which is equivalent to
$\phi$
and in the syntactic closure of $\phi$.
We will make tacit use of transformation $(\alpha)$.
Observe that
it suffices to show that either such a CNF or such a DNF can be
computed,
since one can convert between such a CNF and such a DNF
by use of transformation $(\delta)$.
The procedure is defined as follows.
\begin{itemize}

\item When $\phi$ is an atom or a negated atom, the procedure returns $\phi$.

\item When $\phi$ has the form $\psi_1 \wedge \psi_2$, the procedure
  computes
a CNF for $\phi$ by taking the conjunction of CNFs for $\psi_1$ and
$\psi_2$,
which can be computed recursively.
The case where $\phi$ has the form $\psi_1 \vee \psi_2$ is defined
dually.

\item When $\phi$ has the form $\exists x \psi$, the procedure
  performs the following.
(The case where $\phi$ has the form $\forall y \psi$ is dual.)
First, it recursively computes a DNF for $\psi$; denote the DNF by
$\bigvee_i \psi_i$.  By appeal to transformation $(\beta)$,
the formula $\phi$ is logically equivalent to $\bigvee_i (\exists x \psi_i)$.
To obtain a DNF for $\phi$, it thus suffices to show that 
each formula of the form $\exists x \psi_i$ can be transformed to
an equivalent conjunction of organized formulas.
Each formula $\psi_i$ is a conjunction 
$\psi_i^1 \wedge \cdots \wedge \psi_i^k$ of organized formulas; 
write $\psi_i$ as
$\psi_i^x \wedge \psi_i^{-x}$, where $\psi_i^x$ is the conjunction of
the formulas $\psi_i^j$ such that $x \in \free(\psi_i^j)$,
and $\psi_i^{-x}$ is the conjunction of the remaining formulas $\psi_i^j$.
We have, by transformation $(\gamma)$, that
$\exists x \psi_i$ is equivalent to $(\exists x \psi_i^x) \wedge \psi_i^{-x}$.

\end{itemize}

The algorithm $\org$, given a formula $\phi$, computes an equivalent formula
$\phi'$ where negations appear only in front of atoms
(using transformations $(\epsilon)$),
and then invokes the described procedure on $\phi'$ and outputs either the CNF
or the DNF returned by the procedure.
\end{proof}

In what follows, when $V$ is a set of variables and 
$Q \in \{ \exists, \forall \}$
is a quantifier, we will use 
$Q V$ as shorthand for $Q v_1 \ldots Q v_n$, where $v_1, \ldots, v_n$
is a list of the elements of $V$.  
Our discussion will always be independent of the particular ordering chosen.
Relative to a hypergraph $H$ and a subset $S \subseteq V(H)$,
we consider a set of edges $\{ e_1, \ldots, e_k \}$ to be 
\emph{$S$-connected} if 
one has connectedness of the graph with vertices $\{ e_1, \ldots, e_k
\}$
and having an edge between $e_i$ and $e_j$ if and only if
$S \cap e_i \cap e_j \neq \emptyset$; we say that the hypergraph $H$
is itself \emph{$S$-connected} if $E(H)$ is $S$-connected.

\begin{definition}
We define the sets of 
\emph{$\exists$-layered formulas} 
and of \emph{$\forall$-layered formulas}
to be the variable-loose formulas
that can be constructed inductively, as follows:
\begin{itemize}

\item Each atom and each negated atom is both
an $\exists$-layered formula and a $\forall$-layered formula.

\item If each of $\phi_1, \ldots, \phi_n$ is a $\forall$-layered
  formula, and\\
$X \subseteq \free(\phi_1) \cup \cdots \cup \free(\phi_n)$ is such that
the hypergraph $\{ \free(\phi_1), \ldots, \free(\phi_n) \}$
is $X$-connected,
then 
$\exists X (\bigwedge_{i=1}^n \phi_i)$
is an $\exists$-layered formula.

\item If each of $\phi_1, \ldots, \phi_n$ is an $\exists$-layered
  formula, and\\
$Y \subseteq \free(\phi_1) \cup \cdots \cup \free(\phi_n)$ is such that
the hypergraph $\{ \free(\phi_1), \ldots, \free(\phi_n) \}$
is $Y$-connected,
then 
$\forall Y (\bigvee_{i=1}^n \phi_i)$
is a $\forall$-layered formula.

\end{itemize}
\end{definition}

\begin{theorem}
\label{thm:organized-to-layered}
There exists an algorithm that, given an organized formula $\phi$ as
input,
outputs a layered formula
that is logically equivalent to $\phi$.
\end{theorem}

The intuitive idea behind the algorithm of
Theorem~\ref{thm:organized-to-layered}
is to combine together (into a set quantification $Q V$) quantifiers of the same type
that occur adjacently in $\phi$.


\begin{proof}
We define the algorithm recursively.
\begin{itemize}

\item If $\phi$ is an atom or a negated atom, the algorithm returns
  $\phi$.

\item 
If $\phi$ is an organized formula 
of the form
$\exists v (\bigwedge_{i=1}^{\ell} \theta_i)$ with
$v \in \free(\theta_1) \cap \cdots \cap \free(\theta_{\ell})$, the algorithm
proceeds as follows.
By renaming variables if necessary, the algorithm ensures that no variable is
quantified twice in $\phi$, and also that no variable is both free and
quantified in $\phi$.
For each $i$, the algorithm recursively computes, from $\theta_i$,
a logically equivalent layered formula $\theta'_i$.
The algorithm then writes $\phi$ as the formula
$$\exists v ((\exists X_1 \phi_1) \wedge \cdots \wedge (\exists X_m
\phi_m)) \wedge (\psi_1 \wedge \cdots \wedge \psi_n)$$
where 
$\exists X_1 \phi_1, \ldots, \exists X_m \phi_m$ is a list of the
formulas
$\theta'_i$ that begin with existential quantification,
and where $\psi_1 \wedge \cdots \wedge \psi_n$ is a list of the
remaining formulas $\theta'_i$.
Note that 
each $\phi_j$ is the conjunction of $\forall$-layered formulas, and 
each $\psi_j$ is a $\forall$-layered formula.

For each $j \in \und{m}$, let $H_j$ denote the hypergraph of $\phi_j$,
that is, the hypergraph where an edge is present if it is
the set of free variables of a conjunct of $\phi_j$.
We have that the hypergraph $H_j$ is $X_j$-connected.
Since each $H_i$ contains $v$ in an edge and it holds that
$v \in \free(\psi_1) \cap \cdots \cap \free(\psi_n)$,
we have that the hypergraph $H$ with edge set
$$E(H_1) \cup \cdots \cup E(H_m) \cup 
\{ \free(\psi_1) \} \cup \cdots \cup \{ \free(\psi_n) \}$$
is $(X_1 \cup \cdots \cup X_m \cup \{ v \})$-connected.
Consider the conjunction of the $\phi_j$ and of the $\psi_j$.
This conjunction
can be viewed as the conjunction of
$\forall$-layered formulas whose free variable sets are exactly the
edges of $H$; denote this conjunction by $\chi$.
The $\exists$-layered formula 
$\exists (X_1 \cup \cdots \cup X_m \cup \{ v \}) \chi$
is logically equivalent to $\phi$; this is because, by the variable
renaming
done initially,
the variables in a set $X_i$ do not appear in a formula 
$\phi_j$ when $j \neq i$,
nor do they appear in a formula $\psi_j$.

\item 
If $\phi$ is an organized formula 
of the form
$\forall v (\bigvee_{i=1}^{\ell} \theta_i)$ with
$v \in \free(\theta_1) \cap \cdots \cap \free(\theta_{\ell})$, the algorithm
proceeds dually to the previous case.

\end{itemize}
\end{proof}

\begin{theorem}
\label{thm:formula-to-pos-comb-layered}
There exists an algorithm $\lay$ that, given a formula $\phi$ as input,
outputs a positive combination $\lay(\phi)$ of layered formulas 
that is logically equivalent to $\phi$.
\end{theorem}

\begin{proof}
Immediate from Theorems~\ref{thm:formula-to-pos-comb-organized}
and~\ref{thm:organized-to-layered}.
\end{proof}

\begin{definition} 
We define the following measures on layered formulas.
\begin{itemize}

\item When $\phi$ is an atom or a negated atom,
we define $\thickl(\phi) = |\free(\phi)|$.

\item Suppose that $\phi$ is a layered formula of the form
$\exists U (\bigwedge_{i=1}^n \phi_i)$ or
$\forall U (\bigvee_{i=1}^n \phi_i)$.

We define the \emph{local thickness} of $\phi$ as 
%
$$\localthickl(\phi) = 1 + \tw(\{ \free(\phi_i)  ~|~ i \in \und{n} \} \cup \{ \free(\phi) \})$$
%
where the object to which the treewidth is applied is a hypergraph
specified by its edge set, which hypergraph has vertex set
$\cup_{i=1}^n \free(\phi_i)$.

We define the \emph{thickness} of $\phi$ inductively as
$$\thickl(\phi) = \max (\{\localthickl(\phi) \} \cup \{ \thickl(\phi_i)
~|~ i \in \und{n} \}).$$

We define the \emph{quantified thickness} of $\phi$ as
$$\quantthickl(\phi) = 1 + \tw(\{ \free(\phi_i) \cap U ~|~ i \in
\und{n} \}).$$
\end{itemize}
\end{definition}

\begin{definition}
The \emph{thickness} of an arbitrary formula $\phi$ is defined as
follows:
let $\Psi$ be the set of layered formulas 
that are positively combined subformulas of
$\lay(\phi)$
such that $\lay(\phi)$ is a positive combination over
$\Psi$; then,
$$\thick(\phi) = \max_{\psi \in \Psi} \thickl(\psi).$$
\end{definition}

\begin{prop}
\label{prop:thickness-computable}
The function $\thick(\cdot)$ is computable.
\end{prop}

\begin{proof}
This follows from the computability of $\lay(\phi)$ from $\phi$
(Theorem~\ref{thm:formula-to-pos-comb-layered})
and the definition of thickness.
\end{proof}

We indeed now demonstrate a principal property of thickness,
namely, that this measure provides an upper bound on 
the number of variables needed to express a formula;
this upper bound is effective in that there is an algorithm 
that computes an equivalent formula using a bounded number of
variables.
When $k \geq 1$,
let us say that a formula \emph{uses $k$ many variables}
if the set containing all variables
that occur in the formula has size
less than or equal to $k$.

\begin{theorem}
\label{thm:thick-many-variables}
There exists an algorithm that,
given as input a formula $\phi$, 
outputs an equivalent formula that uses $\thick(\phi)$ many variables.
\end{theorem}

\begin{definition}
\label{def:elim-ordering-of-formula}
Let $\phi$ be a formula of the form
$\exists V (\bigwedge_{i=1}^n \phi_i)$ or
$\forall V (\bigvee_{i=1}^n \phi_i)$.
Let us say that a pair
$e = ((v_1, \ldots, v_m), E)$
consisting of an ordering $v_1, \ldots, v_m$ of the elements
in $\bigcup_{i=1}^n \free(\phi_i)$ 
and a subset $E \subseteq K( \{ v_1, \ldots, v_m \} )$
is an \emph{elimination ordering of $\phi$} if:
\begin{itemize}

\item 
the variables in $\free(\phi)$ occur first in the ordering, that is,
$\{ v_1, \ldots, v_{|\free(\phi)|} \} = \free(\phi)$; and,

\item $e$ is an elimination ordering of the hypergraph
$\{ \free(\phi) \} \cup \{ \free(\phi_i) ~|~ i \in \und{n} \}$.

\end{itemize}
\end{definition}

The following lemma can be taken as 
a variation of a lemma of 
Kolaitis and Vardi~\cite[Lemma 5.2]{KolaitisVardi00-containment}.

\begin{lemma}
\label{lemma:formula-and-elimination-ord}
There exists a polynomial-time algorithm that, 
given a formula $\phi$ of the form
$\exists V (\bigwedge_{i=1}^n \phi_i)$ or
$\forall V (\bigvee_{i=1}^n \phi_i)$
and an elimination ordering $e$ of $\phi$,
outputs a formula $\phi'$ that is logically equivalent to $\phi$
such that
$\width(\phi') \leq 
\max (\{ 1 + \lowerdeg(e) \} \cup \{ \width(\phi_i) ~|~ i \in \und{n} \})$.
\end{lemma}

\begin{proof}
Given a formula $\phi$ of the described form 
having $|V| \geq 1$
and an elimination ordering $e = ((v_1, \ldots, v_m), E)$
of it,
we explain how eliminate the last variable;
precisely speaking, 
we explain how to compute a logically equivalent $\psi$
of the form $\exists \{ v_1, \ldots, v_{m-1} \} \bigwedge \psi_j$
and having elimination ordering
$((v_1, \ldots, v_{m-1}), E \cap K(\{ v_1, \ldots, v_{m-1}\}))$.
The desired algorithm iterates this variable elimination.

Let $\phi^{v_m}$ denote the conjunction of the formulas of the 
form $\phi_i$ having $v_m \in \free(\phi_i)$,
and let $I$ denote the set containing 
the indices of the remaining formulas $\phi_i$.
The formula $\psi$ is defined as 
$\exists \{ v_1, \ldots, v_{m-1} \}$ followed by 
the conjunction of all formulas in 
$\Psi = \{ \phi_i ~|~ i \in I \} \cup \{ \exists v_m \phi^{v_m}\}$.
It is clear that $\psi$ and $\phi$ are logically equivalent,
and we have $|\free(\phi^{v_m})| = 1 + \lowerdeg(e, v_m)$.
We verify that
$e^{\psi} = ((v_1, \ldots, v_{m-1}), E \cap K(\{ v_1, \ldots, v_{m-1}\}))$
is an elimination ordering of $\psi$ as follows.
Since $\free(\psi) = \free(\phi)$, it is clear that the variables
in $\free(\psi)$ occur first in the ordering
$(v_1, \ldots, v_{m-1})$.
To verify that $e^{\psi}$ is an elimination ordering
of the hypergraph named in
Definition~\ref{def:elim-ordering-of-formula},
we need to verify that, for each edge $f$ of that hypergraph,
one has the containment $K(f) \subseteq E$.  For each such edge $f$,
other than the edge $\free(\exists v_m \phi^{v_m})$,
this containment holds because it held for the
original elimination ordering for $\phi$.
For the edge $\free(\exists v_m \phi^{v_m})$,
the containment holds due to the lower neighbor property
and due to $v_m$ being a neighbor of each element of
$\free(\phi^{v_m}) \setminus \{ v_m \}$ in $E$.

We now confirm 
that we can apply the algorithm to $\psi$ to obtain the desired formula
$\phi'$.
We have 
$\lowerdeg(e^{\psi}) \leq \lowerdeg(e)$ 
and 
$\width(\phi^{v_m}) \leq 
\max( \{ \width(\phi_i) ~|~ i \in \und{n} \setminus I \} \cup 
  \{ 1 + \lowerdeg(e) \})$.
From this, it follows that 
$$\max (\{ 1 + \lowerdeg(e^{\psi}) \} \cup \{ \width(\psi_j) ~|~
\psi_j \in \Psi \})$$
$$\leq
\max (\{ 1 + \lowerdeg(e) \} \cup \{ \width(\phi_i) ~|~ i \in \und{n} \}).$$
\end{proof}

\begin{proof}
(Theorem~\ref{thm:thick-many-variables})
By definition of $\thick(\phi)$ and by the computability of
$\lay(\phi)$ from $\phi$ (Theorem~\ref{thm:formula-to-pos-comb-layered}), 
it suffices to prove that there is an algorithm that,
given a layered formula $\phi$, returns an equivalent formula
that uses $\thickl(\phi)$ many variables.

Consider the following algorithm $A$ defined recursively on 
layered formulas.
When $\phi$ is an atom or a negated atom,
$A(\phi) = \phi$.
When $\phi$ is of the form
$\exists V (\bigwedge_{i=1}^n \phi_i)$ or
$\forall V (\bigvee_{i=1}^n \phi_i)$,
the algorithm proceeds as follows.
Set $H$ to be the hypergraph
$$\{ \free(\phi_i)  ~|~ i \in
\und{n} \} \cup \{ \free(\phi) \}.$$
We have $1 + \tw(H) = \localthickl(\phi) \leq \thickl(\phi)$.
By calling the algorithm of
Proposition~\ref{prop:tw-gives-elimination-ordering}
on the hypergraph $H$ and the distinguished edge $\free(\phi)$,
we can obtain
an elimination ordering $e$ of $\phi$
having $1 + \lowerdeg(e) = 1 + \tw(H) \leq \thickl(\phi)$.
Let $\phi'$ be the formula obtained
from $\phi$ by replacing each formula $\phi_i$ by $A(\phi_i)$;
for each $i \in \und{n}$, we have $\width(\phi_i) \leq \thickl(\phi_i)$.
By applying the algorithm of Lemma~\ref{lemma:formula-and-elimination-ord} 
to $\phi'$ and $e$,
we obtain a formula $\phi''$ having
$\width(\phi'') \leq \thickl(\phi)$.
By renaming variables in $\phi''$, 
we can obtain
an equivalent formula
where the number of variables used is equal to $\width(\phi'')$.
The output $A(\phi)$ of the algorithm is this equivalent formula.
\end{proof}

\section{Graph-like queries}
\label{sect:main}

In this section, we state our main theorem 
and two corollaries thereof, which describe the tractable
graph-like sentence sets.
In the following two sections, we prove 
the hardness portion of the main theorem.

\begin{definition}
Define $\mc$ to be the language of pairs
$$\{ (\phi, \relb) ~|~ \relb \models \phi \}$$
where $\phi$ denotes a first-order sentence, and $\relb$
denotes a relational structure.
\end{definition}

\begin{theorem} (Main theorem)
\label{thm:main}
Let $\Phi$ be a graph-like set of sentences having bounded arity.
If $\Phi$ has bounded thickness, then 
the case problem $\mc[\Phi]$ is in $\case{\fpt}$;
otherwise, the case problem $\mc[\Phi]$ is 
$\case{\wone}$-hard or $\case{\cowone}$-hard.
\end{theorem}

We give a proof of this theorem that makes a single forward reference to
the main theorem of Section~\ref{sect:hardness}.

\begin{proof}
Suppose $\Phi$ has thickness bounded above by $k$.
Define $S'$ to be the set of sentences having
thickness less than or equal to $k$.
The set $S'$ is computable by
Proposition~\ref{prop:thickness-computable}.
We have that $\param{\mc[S']}$ is in $\fpt$
via the algorithm that, given an instance
$(\phi, \relb)$, first checks if $\phi \in S'$, and if so,
invokes Theorem~\ref{thm:thick-many-variables} to obtain $\phi'$,
and then performs the natural bottom-up, polynomial-time evaluation
of $\phi'$ on $\relb$ (\`{a} la Vardi~\cite{Vardi95-boundedvariable}).
The case problem $\mc[\Phi]$ slice reduces to
$\mc[S']$ via the slice reduction
$(\{ (s', s') ~|~ s' \in S' \}, \pi_3)$.

Suppose that $\Phi$ has unbounded thickness.
Theorem~\ref{thm:hardness} yields that
either $\caseclique$ or $\casecoclique$ slice reduces to
$\mc[\Phi]$.
It then follows by definition that
$\mc[\Phi]$ is 
$\case{\wone}$-hard or $\case{\cowone}$-hard,
since  the problems
$\clique$ and $\coclique$ are $\wone$-hard and $\cowone$-hard,
respectively.
\end{proof}

We now provide two corollaries that describe the complexity
of the problems $\param{\mc[\Phi]}$ addressed by the main theorem;
the first corollary assumes that $\Phi$ is computable,
while the second corollary makes no computability assumption on $\Phi$.
Both of these corollaries follow directly from Theorem~\ref{thm:main}
via use of Propositions~\ref{prop:in-casec} and~\ref{prop:casec-hard}.

\begin{corollary}
\label{cor:main-computable}
Let $\Phi$ be a computable, graph-like set of sentences having bounded arity.
If $\Phi$ has bounded thickness,
then the problem 
$\param{\mc[\Phi]}$ 
is in $\fpt$;
otherwise, the problem 
$\param{\mc[\Phi]}$ 
is not in $\fpt$, unless $\wone \subseteq \fpt$.
\end{corollary}


\begin{proof}
When $\Phi$ has bounded thickness, the claim follows directly from
Theorem~\ref{thm:main} and Proposition~\ref{prop:in-casec}.
When $\Phi$ does not have bounded thickness,
Theorem~\ref{thm:main} and Proposition~\ref{prop:casec-hard}
imply that $\param{\mc[\Phi]}$ is either
$\wone$-hard or $\cowone$-hard, from which the claim follows.
\end{proof}

\begin{corollary}
\label{cor:main-noncomputable}
Let $\Phi$ be a graph-like set of sentences having bounded arity.
If $\Phi$ has bounded thickness,
then the problem
$\param{\mc[\Phi]}$ 
is in $\nu{\fpt}$;
otherwise, the problem 
$\param{\mc[\Phi]}$ 
is not in $\nu{\fpt}$, unless $\wone \subseteq \nu{\fpt}$.
\end{corollary}


\begin{proof}
When $\Phi$ has bounded thickness, the claim follows directly from
Theorem~\ref{thm:main} and Proposition~\ref{prop:in-casec}.
When $\Phi$ does not have bounded thickness,
Theorem~\ref{thm:main} and Proposition~\ref{prop:casec-hard}
imply that $\param{\mc[\Phi]}$ is either
non-uniformly $\wone$-hard or non-uniformly $\cowone$-hard,
so, invoking the fact that $\nu{\fpt}$ is closed under
nu-FPT-reductions,
the claim follows.
\end{proof}

In Section~\ref{sect:discussion},
we provide a discussion
of how the main theorem of the present paper can be used
to readily derive the dichotomy theorem of graphical sets of
quantified conjunctive
queries~\cite{ChenDalmau12-decomposingquantified};
a dual argument yields the corresponding theorem on graphical sets
of quantified disjunctive queries.
Note that it follows immediately from this discussion and
\cite[Example 3.5]{ChenDalmau12-decomposingquantified}
that there exists a set of graph-like sentences having
bounded arity
that is tractable, but not contained in one of the tractable classes
identified by Adler and Weyer~\cite{AdlerWeyer12-treewidth}.

Let us also note here that, as pointed out by Adler and Weyer,
it is undecidable, given a first-order sentence $\phi$
and a value $k \geq 1$, 
whether or not $\phi$ is logically equivalent to a $k$-variable
sentence,
and hence one cannot expect an algorithm
that takes a first-order sentence and outputs an equivalent one
that minimizes the number of variables.
(This undecidability result holds even for 
positive first-order logic~\cite{BovaChen14-width-ep}.)

\section{Accordion reductions}
\label{sect:accordion}

In this section, we introduce a notion that we call 
\emph{accordion reduction}.
When $C \subseteq \Sigma^* \times \Sigma^*$ is a set of string pairs
and $S \subseteq \Sigma^*$ is a set of strings, we use
$\clo_C(S)$ to denote the intersection of all sets $T$ containing $S$
and
having the closure property that,
if $(u, u') \in C$ and $u' \in T$, then $u \in T$.
When an accordion reduction exists for such a $C$
(with respect to a language $Q$),
the theorem of this section 
(Theorem~\ref{thm:accordion-reduction-gives-slice-reduction})
yields that
the problem $Q[\clo_C(S)]$ slice reduces to $Q[S]$.
Hence, an accordion reduction is not itself a slice reduction,
but its existence provides a sufficient condition
for the existence of a class of slice reductions.

How is this section's theorem proved?
One component of an accordion reduction is 
an FPT-computable mapping $r$ that,
for each $(u, u') \in C$,
maps a $Q$-instance $(u, y)$ to a $Q$-instance $(u', y')$.
Intuitively, to give a slice reduction from
$Q[\clo_C(S)]$ to $Q[S]$,
one needs to reduce, for any $s \in S$ and
$s_1 \in \clo_C(S)$,
instances of the form $(s_1, \cdot)$ to
instances of the form $(s, \cdot)$.
The containment $s_1 \in \clo_C(S)$
implies the existence of a sequence
$s_1, \ldots, s_k = s$ such that every pair
$(s_i, s_{i+1})$ is in $C$.
This naturally suggests applying the map $r$
repeatedly, but note that there is no constant bound
on the length $k$ of the sequence.
Hence, $r$ needs to be sufficiently well-behaved so that,
when composed with itself arbitrarily many times
(in the described way), 
the end effect is that of an FPT-computable function
that may serve as the map in the definition of slice reduction.
(The author is mentally reminded of the closing of an accordion in 
thinking that this potentially long sequence of compositions
yields a single well-behaved map.)
To ensure this well-behavedness,
we impose a condition that we call \emph{measure-linearity}.

In the context of accordion reductions,
a \emph{measure} is a mapping $m: \Sigma^* \times \Sigma^* \to
\N$
such that 
there exist a computable function $f: \Sigma^* \to \N$
and a polynomial $p: \N \to \N$ 
whereby,
for all pairs $(s, y) \in \Sigma^* \times \Sigma^*$,
it holds that 
$m(s,y) \leq |y| \leq f(s) p(m(s,y))$.

\begin{definition}
Let $Q \subseteq \Sigma^* \times \Sigma^*$ be a language of pairs.
With respect to $Q$, an \emph{accordion reduction}
consists of:
\begin{itemize}

\item a computably enumerable language $C \subseteq \Sigma^* \times \Sigma^*$,

\item a measure $m: \Sigma^* \times \Sigma^* \to \N$, 

\item a partial computable function 
$B: \Sigma^* \times \Sigma^* \to \N$
with $\dom(B) = C$, and

\item a mapping $r: \Sigma^* \times \Sigma^* \times \Sigma^* \to
  \Sigma^*$
that has $\dom(r) = C \times \Sigma^*$,
that is FPT-computable with respect to 
the parameterization $(\pi_1, \pi_2)$,
and that is \emph{measure-linear} in that,
for each pair $(u, u') \in C$
and
for each $y \in \Sigma^*$, it holds that
$m(u', r(u, u', y)) \leq B(u, u') m(u, y)$,

\end{itemize}
such that the following condition holds:
\begin{itemize}

\item \emph{(correctness)} 
for each $(u, u') \in C$, it holds
(for each $y \in \Sigma^*$) that
$$(u, y) \in Q \Leftrightarrow (u', r(u, u', y)) \in Q.$$

\end{itemize}
\end{definition}

\begin{theorem}
\label{thm:accordion-reduction-gives-slice-reduction}
Suppose that $Q[S]$ is a case problem,
and that $(C, m, r)$ is an accordion reduction with respect to $Q$.
Then, the case problem $Q[\clo_C(S)]$ 
slice reduces to the case
problem $Q[S]$.
\end{theorem}


\begin{proof}
We define a slice reduction $(U, r^+)$ from
$Q[\clo_C(S)]$ to $Q[S]$.
Set $U$ to be the set containing the pairs $(u, u')$
such that 
there exists $k \geq 1$ and a sequence $u_1, \ldots, u_k$
such that $u = u_1$, $u_k = u'$, and
$(u_i, u_{i+1}) \in C$ when $1 \leq i < k$.
It is straightforwardly verified that
$U$ is computably enumerable and that
the coverage criterion is satisfied, that is, 
if $u \in \clo_C(S)$, then there exists $u' \in S$
such that $(u, u') \in U$.
Fix $A_U$ to be an algorithm that, given a pair $(u, u')$,
returns a sequence $u = u_1, u_2, \ldots, u_k = u'$ of the
just-described form whenever $(u, u') \in U$.
Consider the algorithm $A_{r^+}$ that
does the following: given a triple $(u, u', y)$, it
invokes $A_U(u, u')$; if this computation halts
with output $u = u_1, u_2, \ldots, u_k = u'$,
the algorithm sets $y_1 = y$;
the algorithm computes
\begin{center}
$y_2 = r(u_1, u_2, y_1)$,
$y_3 = r(u_2, u_3, y_2)$, $\ldots$,
$y_k = r(u_{k-1}, u_k, y_{k-1})$,
\end{center}
and outputs $y_k$.
We define $r^+$ as the partial mapping computed by this algorithm $A_{r^+}$.
In order to compute each of the strings $y_2, \ldots, y_k$, the
algorithm
$A_{r^+}$
uses the algorithm $A_r$ for $r$ provided by the definition 
of an accordion reduction; let $f_r$ and $p_r$ be a computable function and 
a polynomial, respectively, such that the running time of $A_r$,
on an input $(u, u', y)$,
is bounded above by $f_r(u, u') p_r(|(u, u', y)|)$.

On an input $(u, u', y)$ where $r^+$ is defined, 
the running time of $A_{r^+}$ can be bounded above by 
the running time of $A_U(u, u')$ plus
$$f_r(u_1, u_2) p_r(|(u_1, u_2, y_1)|) + \cdots
+ f_r(u_{k-1}, u_k) p_r(|(u_{k-1}, u_k, y_{k-1})|).$$
Note that the running time of $A_U(u, u')$ and $k$ are both functions
of $(u, u')$.
So, in order to show that the running time of
$A_{r^+}(u, u', y)$ is FPT-computable with respect to $(\pi_1,
\pi_2)$,
it suffices to show that, for each $i \geq 1$, the term
$f_r(u_i, u_{i+1}) p_r(|(u_i, u_{i+1}, y_i)|)$ is degree-bounded, when defined
(hereon, when discussing degree-boundedness,
this is with respect to $(\pi_1, \pi_2)$).
So consider such a term; the strings $u_i$ and $u_{i+1}$ 
are computable functions of $(u, u')$,
and so
by appeal to Proposition~\ref{prop:db-closure}, it suffices to show that 
the string length $|y_i|$ is degree-bounded.
This is clear in the case that $i = 1$, that is,
we have that $|y_1|$ is degree-bounded.
Let us consider the case that $i > 1$.
By the measure-linearity of $r$, we have that
$$m(u_i, y_i) \leq B(u_1, u_2) \cdots B(u_{i-1}, u_i) m(u_1,y_1).$$
By the definition of measure, it follows that
$$m(u_i, y_i) \leq B(u_1, u_2) \cdots B(u_{i-1}, u_i) |y_1|.$$
Since the constants $B(u_j, u_{j+1})$ depend only on $(u, u')$,
we obtain that the value 
$m(u_i, y_i)$ is degree-bounded.
Letting $f_m$ and $p_m$ be the computable function 
and polynomial  (respectively)
provided
by the definition of measure,
we then have $|y_i| \leq f_m(u_i) p_m(m(u_i, y_i))$,
and conclude that $|y_i|$ is degree-bounded.
\end{proof}

\section{Hardness}
\label{sect:hardness}

\newcommand{\full}{\mathsf{full}}
\newcommand{\sdot}{\circ^s}

In this section, we establish the main intractability result of the
paper,
namely,
that the case problem $\mc[\Phi]$ is hard when
$\Phi$ is graph-like and has unbounded thickness.

\begin{theorem}
\label{thm:hardness}
Suppose that $\Phi$ is a set of graph-like sentences
of bounded arity such that
$\thick(\Phi)$ is unbounded.
Then, either 
$\caseclique$ or $\casecoclique$ 
slice reduces to $\mc[\Phi]$.
\end{theorem}

\begin{proof}
Immediate from Lemmas~\ref{lemma:red-friendly-to-graphlike}
and~\ref{lemma:red-fullysorted-to-normal},
and Theorem~\ref{thm:red-clique-to-sorted-friendly}.
\end{proof}

This intractability result is obtained by composing three slice reductions.
Define a formula to be \emph{friendly} 
if it is loose, positive, and layered.
We first show
(Lemma~\ref{lemma:red-friendly-to-graphlike})
that there exists a set of friendly sentences $\Psi$,
with $\thickl(\Psi)$ unbounded,
such that $\mc[\Psi]$ slice reduces to $\mc[\Phi]$.
We next show
(Lemma~\ref{lemma:red-fullysorted-to-normal})
that a multi-sorted version $\full(\Psi)$ of $\Psi$
has the property that
$\mcs[\full(\Psi)]$ slice reduces to $\mc[\Psi]$;
here, $\mcs$ denotes the multi-sorted generalization of $\mc$.
Finally, we directly slice reduce either
$\caseclique$ or $\casecoclique$ to $\mcs[\Psi]$
(Theorem~\ref{thm:red-clique-to-sorted-friendly});
this third reduction is obtained via an accordion reduction.

\begin{lemma}
\label{lemma:red-friendly-to-graphlike}
Suppose that $\Phi$ is a set of graph-like sentences
of bounded arity such that
$\thick(\Phi)$ is unbounded.
There exists a set of friendly sentences $\Psi$
such that $\mc[\Psi]$ slice reduces to $\mc[\Phi]$
and such that $\thickl(\Psi)$ is unbounded.
\end{lemma}


\begin{proof}
Define $\Phi'$ to be the set 
that contains a loose layered sentence 
if it occurs as a positively combined subformula of a loose sentence in 
$\lay(\{ \phi \in \Phi  ~|~  \textup{$\phi$ is loose} \})$.
We have that $\thickl(\Phi')$ is unbounded;
this is because, for each sentence $\phi \in \Phi$,
if we define $\phi^L$ to be equal to $\org(\phi)$ but with symbols renamed
(if necessary) so that $\phi^L$ is loose, 
then $\phi^L \in \Phi$ and it is straightforwardly verified
that $\thick(\phi) = \thick(\phi^L)$.

We claim that $\mc[\Phi']$ slice reduces to $\mc[\Phi]$.
We define a slice reduction as follows.
The set $U$ contains a pair $(\phi', \phi)$
if $\phi$ is a loose sentence and
$\phi'$ is a loose layered sentence that is a positively combined subformula
of $\lay(\phi)$; we have that $U$ is computable.
For pairs $(\phi', \phi) \in U$,
we set $r(\phi', \phi, \relb') = \relb$,
where $\relb$ is defined as follows.
For each symbol $R$ (of arity $k$) that appears in $\phi$
but not in $\phi'$, 
we define $R^{\relb} = B^k$ or $R^{\relb} = \emptyset$ as appropriate
so that $\relb' \models \phi'$ if and only if $\relb \models \phi$;
this is possible since $\phi'$ is a positively combined subformula of
$\lay(\phi)$
and because $\lay(\phi)$ and $\phi$ are logically equivalent
(by Theorem~\ref{thm:formula-to-pos-comb-layered}).

Define $\Psi$ to be the set that contains a sentence $\psi$
if it can be obtained from a sentence $\phi' \in \Phi'$
by removing all negations that appear immediately in front of atoms.
Since each sentence in $\Phi'$ is loose and layered,
we obtain that each sentence in $\Psi$ is friendly.
We give a slice reduction from $\mc[\Psi]$ to $\mc[\Phi']$,
as follows.
Define $U$ to be the set that contains a pair $(\psi, \phi')$
if $\phi'$ is a loose layered sentence and
$\psi$ can be obtained from $\phi'$ by removing all negations
that appear in front of atoms.
When $(\psi, \phi') \in U$,
define $r(\psi, \phi', \relb) = \relb'$
where $B' = B$ and, for each symbol $R$ of arity $k$,
it holds that $R^{\relb'} = B^k \setminus R^{\relb}$
when $R$ appears in $\phi'$ with a negation before it,
and $R^{\relb'} = R^{\relb}$ otherwise.
(Note that $B^k \setminus R^{\relb}$ can always be computed
in polynomial time from $R^{\relb}$ due to our assumption of bounded arity.)
From the definition of $\thickl(\cdot)$,
we have that $(\psi, \phi') \in U$ implies $\thickl(\psi) =
\thickl(\phi')$,
so $\thickl(\Psi)$ is unbounded.

By Theorem~\ref{thm:slice-reduction-transitive}, 
there is a slice reduction from
$\mc[\Psi]$ to $\mc[\Phi]$.
\end{proof}

In what follows, we will work with multi-sorted relational first-order
logic, formalized as follows.
(For differentiation, we will refer to formulas in the usual first-order logic
considered thus far as \emph{one-sorted}.)
A \emph{signature} is a pair $(\sigma, \S)$
where $\S$ is a set of \emph{sorts} and $\sigma$
is a set of relation symbols;
each relation symbol $R \in \sigma$ has an associated arity
$\ar(R)$ which is an element of $\S^*$.
In a formula over signature $(\sigma, \S)$,
each variable $v$ has associated with it a sort $s(v)$ from $\S$;
an atom is a formula $R(v_1, \ldots, v_k)$ where $R \in \sigma$
and $s(v_1) \ldots s(v_k) = \ar(R)$.
A structure $\relb$ on signature $(\sigma, \S)$
consists of an $\S$-sorted family $\{ B_s ~|~ s \in \S \}$
of sets called the \emph{universe} of $\relb$, and for each symbol
$R \in \sigma$, an interpretation $R^{\relb} \subseteq B_{\ar(R)}$,
where for a word $w = w_1 \ldots w_k \in \S^*$, we use
$B_w$ to denote the product $B_{w_1} \times \cdots \times B_{w_k}$.
We use $\mc_s$ to denote the multi-sorted version of $\mc$,
that is, it is the language of pairs $(\phi, \relb)$ where $\phi$ is a
sentence and
$\relb$ is a structure
both having the same signature $(\sigma, \S)$, and $\relb \models \phi$.

Suppose that $\phi$ is a multi-sorted friendly formula
on signature $(\sigma, \S)$, and let $V$ be the set
of variables occurring in $\phi$; 
we say that $\phi$ is \emph{fully-sorted} if 
$V \subseteq \S$ and for each $v \in V$, the sort of $v$ is $v$ itself
(that is, $s(v) = v$).
When $\psi$ is a one-sorted friendly formula,
and $V$ is the set of variables that occur in $\psi$,
we use $\full(\psi)$ to denote the natural fully-sorted formula
induced by $\psi$, namely, the formula on signature $(\sigma, V)$
where $\sigma$ contains those symbols occurring in $\psi$, and,
if $R(v_1, \ldots, v_k)$ appears in $\psi$, then
$\ar(R) = v_1 \ldots v_k$ (this is well-defined since $\phi$ is symbol-loose).

\begin{lemma}
\label{lemma:red-fullysorted-to-normal}
Let $\Psi$ be a set of one-sorted friendly sentences.
The set $\full(\Psi)$ of fully-sorted friendly sentences
has the property that $\mcs[\full(\Psi)]$ slice reduces to $\mc[\Psi]$.
\end{lemma}

\begin{proof}
We give a slice reduction.
Define $U$ to be the set that contains each pair of the form
$(\full(\psi), \psi)$.
Suppose that $\psi$ and $\relb$ are over
signature $(\sigma, \S)$.
Define $r(\full(\psi), \psi, \relb) = \relb'$,
where $\relb'$ is defined as follows.
Let $B'$ be a set whose cardinality is
$\max_{s \in \S} |B_s|$.
For each sort $s \in \S$,
fix a map $f_s: B' \to B_s$ that is surjective.
For each relation symbol $R \in \sigma$ of arity $s_1 \ldots s_k$,
define 
$R^{\relb'} = \{ (b'_1, \ldots, b'_k) \in B'^k ~|~ 
(f_{s_1}(b'_1), \ldots, f_{s_k}(b'_k)) \in R^{\relb} \}$.
It is straightforward to prove by induction that,
for all subformulas $\phi$ of $\psi$,
and each assignment $g$ from the set of variables to $B'$,
that
$\relb', g \models \phi$ if and only if
$\relb, f \sdot g \models \full(\phi)$.
Here, $f \sdot g$ denotes
the mapping that sends each variable $v$
to $f_{s(v)}(g(v))$
and
we view $\full(\phi)$ as a formula over the signature $(\sigma, \S)$
of $\psi$.
The correctness of the reduction follows.
\end{proof}

\emph{In the remainder of this section,
we assume that all formulas and structures under discussion are multi-sorted.}


Let us say that a
friendly formula is a \emph{simple formula} if it is of the form
$\exists X (\bigwedge_{i=1}^n \alpha_i)$ 
or of the form
$\forall Y (\bigvee_{i=1}^n \alpha_i)$ 
where each $\alpha_i$ is an atom.
When $\phi$ is a simple formula
where the variables $V$ are those that are quantified initially,
we say that $\psi$ is 
\emph{a sentence based on $\phi$}
if $\psi$ is a simple sentence derivable from $\phi$
by replacing each atom $R(w_1, \ldots, w_k)$
with an atom whose variables are the elements in
$\{ w_1, \ldots, w_k \} \cap V$.

We now present an accordion reduction that will be used to derive our
hardness result.  
When $\Psi$ is a set of friendly sentences,
this accordion reduction will allow  a simple
subformula
$\phi' = \exists V \chi$ 
to simulate a disjunction of atoms on
$\free(\phi')$,
and likewise for a simple subformula
$\phi' = \forall V \chi$ 
to simulate a corresponding conjunction;
this is made precise as follows.
The set $C$ is defined to contain a pair
$(\psi, \phi)$ of friendly sentences if 
there exists a simple subformula 
$\phi' = Q V \chi$ 
of $\phi$
such that one of the following conditions holds:

\begin{enumerate}

\item[(1)]  $\psi$ is a sentence based on $\phi'$.

\item[(2)] $Q = \exists$ and $\psi$ is a friendly sentence obtained from $\phi$
by replacing $\phi'$ with $\bigvee_{i=1}^m E_i(v_{i1}, v_{i2})$,
where the tuples $(v_{i1}, v_{i2})$
are such that
$\{ v_{11}, v_{12} \}, \ldots, \{ v_{m1}, v_{m2} \}$
is a list of the elements in
$K(\free(\phi'))$ and the $E_i$ are relation symbols
(each of which is fresh in that it does not appear elsewhere in $\psi$).

\item[(3)] $Q = \forall$ and $\psi$ is a friendly sentence obtained from $\phi$
by replacing $\phi'$ with $\bigwedge_{i=1}^m E_i(v_{i1}, v_{i2})$
where the tuples $(v_{i1}, v_{i2})$ and the symbols $E_i$ 
are as described in the previous case.

\end{enumerate}

In order to present the accordion reduction, 
we view, without loss of generality, each pair of strings
as a pair $(\phi, \relb)$ where $\phi$ is
a sentence whose encoding includes
the signature $(\sigma, \S)$ over which it is defined;
and $\relb$ is a structure over this signature.

\begin{theorem}
\label{thm:formula-accordion-reduction}
Let $M$ be the measure such that
$M(\phi, \relb)$ is equal to $\max_{s \in \S} |B_s|$
when $\relb$ is a multi-sorted structure 
defined on the signature $(\sigma, \S)$ of $\phi$.
There exists a mapping 
$r: \Sigma^* \times \Sigma^* \times \Sigma^* \to \Sigma^*$ 
such that the triple $(C, M, r)$ is an accordion reduction 
with respect to $\mcs$.
(Here, $C$ is the set defined above.)
\end{theorem}

\begin{proof}
Suppose $(\psi, \phi, \rela)$ is a triple with 
$(\psi, \phi) \in C$ and where $\rela$ is a structure over the
signature of $\psi$.
We define $r(\psi, \phi, \rela)$ to be the structure $\relb$,
defined as follows.

It is straightforward to treat the 
case where there exists a simple subformula $\phi'$
of $\phi$ such that $\psi$ is the sentence based on $\phi'$
(here, we omit discussion of this case).


In the remainder of this proof,
we consider the case that 
there exist a simple subformula $\phi' = \exists V \chi$ 
of $\phi$
such that $\psi$ is a friendly sentence obtained from $\phi$
by replacing $\phi'$ with a disjunction $\psi'$
as in the definition of $C$.
(Dual to this case is the remaining case 
where $\phi' = \forall V \chi$ and $\psi$ is obtained
from $\phi$ by replacing $\phi'$ with a conjunction.)
We will denote the disjunction $\psi'$ by 
$E_1(w_{11}, w_{12}) \vee \cdots \vee E_m(w_{m1}, w_{m2})$.
We define the universe of $\relb$ as follows.
\begin{itemize}

\item For each sort $u$ of $\rela$, we define $B_u = A_u$.

\item For each $v \in V$, define
$B_v = \{ (E_{\ell}, u, a) ~|~ 
\ell \in \und{m}, u \in \{ w_{\ell 1}, w_{\ell 2} \}, a \in A_u
 \}$
\end{itemize}
As $m \leq |V|^2$ and a variable $u$ appearing as the second
coordinate
in an element of a set $B_v$ must be an element of $\free(\phi')$,
We have $M(\phi, \relb) \leq |V|^2 \cdot |\free(\phi')| \cdot M(\psi, \rela)$; 
this confirms measure-linearity of $M$, 
since $|V|$ and $|\free(\phi')|$ are computable
 functions of the pair $(\psi, \phi)$.

For each atom $R(u_1, \ldots, u_k)$ in $\phi$ that occurs outside of
$\phi'$
(equivalently, that also appears in $\psi$),
define $R^{\relb} = R^{\rela}$.

For each atom $R(u_1, \ldots, u_k)$ that occurs in $\phi'$,
define $R^{\relb}$ to contain a tuple 
$(b_1, \ldots, b_k) \in B_{u_1  \ldots u_k}$
if and only if the following two conditions hold:
\begin{itemize}

\item For all $i, j \in \und{k}$, if $u_i \in V$ and $u_j \in V$, 
then $b_i = b_j$.

\item For all $i, j \in \und{k}$, if 
$u_i \in V$, 
$b_i = (E_{\ell}, u, a)$,
and
$u_j \notin V$  (equivalently, $u_j \in \free(\phi')$),
then
\begin{itemize}

\item $u_j = u$ implies $b_j = a$, and

\item $\{ u_j, u \} = \{ w_{\ell 1}, w_{\ell 2} \}$ implies
$\rela, \{ (u, a), (u_j, b_j) \} \models 
E_{\ell}(w_{\ell 1}, w_{\ell 2})$.  (Here, we use a set of pairs to
denote a partial map.)

\end{itemize}
\end{itemize}

To verify that $\rela \models \psi$ if and only if $\relb \models
\phi$,
it suffices to verify that, for any assignment $f$
defined on $\free(\phi') = \free(\psi')$
taking each variable $u$ to an element of $A_u$,
that
$$\rela, f \models \psi' \Leftrightarrow \relb, f \models \phi'.$$

$(\Rightarrow)$: There exists $\ell \in \und{m}$
such that $\rela, f \models E_{\ell}(w_{\ell 1}, w_{\ell 2})$.
Pick $w$ to be a variable in $\{ w_{ \ell 1 }, w_{ \ell 2 } \}$,
and consider the extension $f^+$ of $f$ that sends each variable in
$V$
to $(E_{\ell}, w, f(w))$.
It is straightforward to verify that $\relb, f^+$ satisfies the
conjunction of $\phi'$; we do so as follows.
Suppose that $R(u_1, \ldots, u_k)$ is an atom in this conjunction.
We claim that $(f^+(u_1), \ldots, f^+(u_k)) \in R^{\relb}$.
This tuple clearly satisfies the first condition in the definition
of $R^{\relb}$; to check the second condition,
suppose that $i, j \in \und{k}$ are such that 
$u_i \in V$ and $u_j \notin V$.
We have $f^+(u_i) = (E_{\ell}, w, f(w))$.
If $u_j = w$, then indeed $f^+(u_j) = f(w)$.
If $\{ u_j, w \} = \{ w_{ \ell 1 }, w_{ \ell 2 } \}$,
then $\{ (w, f(w)), (u_j, f^+(u_j)) \}$ is equal to 
$f \restrict \{ w_{ \ell 1}, w_{\ell 2} \}$,
and we have
$\rela, f \restrict \{ w_{ \ell 1}, w_{\ell 2} \} \models
E_{\ell}(w_{\ell 1}, w_{\ell 2})$
by our choice of $\ell$.

$(\Leftarrow)$: 
Suppose that $\relb, f^+$ satisfies the conjunction of $\phi'$.
By the definition of $\relb$ and since $\phi'$ is a layered formula,
$f^+$ maps all variables in $V$ to the same value
$(E_{\ell}, u, a)$.
There exists an atom
$R(u_1, \ldots, u_k)$ in $\phi'$ such that one of its variables $u_j$
has the property that $\{ u_j, u \} = \{ w_{\ell 1}, w_{\ell 2} \}$.
It follows, from the definition of $R^{\relb}$,
that $\rela, f^+ \restrict \{ u, u_j \} \models E_{\ell}(w_{\ell 1}, w_{\ell 2})$.
\end{proof}

\begin{theorem}
\label{thm:red-clique-to-sorted-friendly}
Let $\Psi$ be a set of fully-sorted, friendly sentences
such that $\thickl(\Psi)$ is unbounded.
Then, either 
$\caseclique$ or $\casecoclique$
slice reduces to $\mcs[\Psi]$.
\end{theorem}

Suppose that $\psi$ is a simple friendly formula 
that occurs as a subformula of a formula $\phi$.
We say that $\psi$ is an
\emph{existential $k$-clique} if 
it is of the form $\exists X (\bigwedge_{i=1}^n \alpha_i)$
and there exists a set $V$ of variables, with $|V| \geq k$, that are existentially
quantified
(in $\phi$) such that for each set $\{ v, v' \} \in K(V)$,
there exists an atom $\alpha_i$ with $\{ v, v' \} \subseteq \free(\alpha_i)$.
We define a \emph{universal $k$-clique} dually.

\begin{proof}
Let $C$ be as defined above in the discussion.
By appeal to Theorem~\ref{thm:formula-accordion-reduction},
it suffices to prove that either 
$\caseclique$ or $\casecoclique$
slice reduces to
$\mcs[\clo_C(\Psi)]$.
When $\Theta$ is a set of layered sentences,
let us use the term \emph{$\Theta$-relevant subformula}
to refer to a layered subformula $\phi$ of a sentence $\theta$ in
$\Theta$.
For each $\Psi$-relevant subformula $\phi$,
it holds that 
$\localthickl(\phi) \leq \quantthickl(\phi) + |\free(\phi)|$
(this is since the hypergraph from which $\quantthickl(\phi)$ is defined
is the hypergraph  from which $\localthickl(\phi)$ is defined,
but with the vertices in $\free(\phi)$ removed).
By the definition of $\thickl(\cdot)$, 
the quantity $\localthickl(\phi)$ over $\Psi$-relevant subformulas $\phi$ is unbounded.
Thus, 
over $\Psi$-relevant subformulas $\phi$,
either the quantity $\quantthickl(\phi)$ or the quantity
$|\free(\phi)|$
is unbounded.
We may therefore consider two cases.

Assume that $|\free(\phi)|$ is unbounded 
over $\Psi$-relevant subformulas $\phi$.
By conditions (2) and (3) in the definition of $C$,
we obtain that $|\free(\phi)|$ is unbounded 
over simple $\clo_C(\Psi)$-relevant subformulas $\phi$.
We assume that $|\free(\phi)|$ is unbounded over such simple subformulas
$\phi$ that use existential quantification
(if not, then $|\free(\phi)|$ is unbounded over such simple
subformulas 
$\phi$ that use universal quantification, and the argumentation is dual).
We now consider two cases.
If the number of universally quantified variables in $\free(\phi)$
is unbounded over such subformulas $\phi$,
then by condition (2) of the definition of $C$ applied to 
each such subformula (and the sentence in which it appears),
we obtain that,
for each $k \geq 1$,
the set $\clo_C(\Psi)$ contains a sentence
that contains a universal $k$-clique.
Otherwise, the number of existentially quantified variables in
$\free(\phi)$
is unbounded over such subformulas $\phi$;
in this case,
by an application of condition (3) followed by an application of
condition (2) (again, to each such subformula $\phi$ and the sentence
in which it appears),
we obtain that,
for each $k \geq 1$,
the set $\clo_C(\Psi)$ contains a sentence
that contains an existential $k$-clique.
In this latter case, let us explain how to exhibit a slice reduction
from $\caseclique$  to $\mcs[\clo_C(\Psi)]$
(in the former case, one dually obtains a reduction 
from $\casecoclique$ to $\mcs[\clo_C(\Psi)]$).
We define $U$ to be the set of pairs $(k, \theta)$ 
such that $k \geq 1$ and $\theta$ is a sentence that contains an
existential $k$-clique as a subformula.
Let us define $r(k, \theta, (V, E))$ to be the structure $\relb$
described
as follows.
Let $W$ be the set of variables 
that witnesses the existential $k$-clique.
For each relation symbol $R$ of an atom in $\theta$ that
does not witness the existential $k$-clique,
we may set $R^{\relb}$ to either $\emptyset$ or $B^r$ 
(here, $r$ is the arity of $R$)
in such a way that
$\relb \models \theta$ if and only if 
$\relb \models \exists W (\wedge_j \beta_j)$,
where the $\beta_j$ are the atoms witnessing the existential
$k$-clique.
By defining, for each remaining relation symbol $F$
(that is, for each relation symbol $F$ of an atom $\beta_j$),
the relation $F^{\relb}$ to be $E$,
we obtain that $(V, E)$ contains a $k$-clique if and only if
$\relb \models \theta$.

Now assume that $\quantthickl(\phi)$ is unbounded
over $\Psi$-relevant subformulas $\phi$.
By conditions (2) and (3) in the definition of $C$,
we obtain that $\quantthickl(\phi)$ is unbounded 
over simple $\clo_C(\Psi)$-relevant subformulas $\phi$.
By considering the sentences based on these subformulas
(and by invoking condition (1) in the definition of $C$),
we obtain that $\clo_C(\Psi)$ contains simple sentences
$\exists X (\bigwedge \alpha_i)$ 
or contains simple sentences
$\exists Y (\bigvee \alpha_i)$ 
such that, over these sentences, $\quantthickl(\cdot)$
is unbounded.
In the former case,
the intractability result of 
Grohe, Schwentick and Segoufin~\cite{GroheSchwentickSegoufin01-conjunctivequeries}
directly yields a slice reduction from 
$\caseclique$  to $\mcs[\clo_C(\Psi)]$
where 
the set $U$ contains a pair $(k, \theta)$ when $k \geq 1$
and $\theta$ is a sentence
$\exists X (\bigwedge \alpha_i)$ 
whose corresponding graph admits a $k$-by-$k$ grid as a minor.
The latter case is dual (one obtains
a slice reduction from $\casecoclique$  to $\mcs[\clo_C(\Psi)]$).
\end{proof}

\newpage

\appendix

\section{Proof of Theorem~\ref{thm:slice-reduction-transitive}}

\begin{proof}
Let $(U_1, r_1)$ be a slice reduction from 
$Q_1[S_1]$ to $Q_2[S_2]$, and
let $(U_2, r_2)$ be a slice reduction from
$Q_2[S_2]$ to $Q_3[S_3]$.
Define $U$ to be the set
\begin{center}
$\{ (t_1, t_3) ~|~ \textup{ there exists $t_2 \in \Sigma^*$ such that}$

$\textup{$(t_1, t_2)
  \in U_1$ and $(t_2, t_3) \in U_2$} \}.$
\end{center}
It is straightforward to verify that $U$ is computably enumerable
and that there exists an algorithm $A_U$ that, given a pair
$(t_1, t_3) \in U$, outputs a value $t_2 \in \Sigma^*$ 
such that $(t_1, t_2)  \in U_1$ and $(t_2, t_3) \in U_2$.
Define $r$ to be a partial function such that, when
$(t_1, t_3) \in U$,
for each $y \in \Sigma^*$ it holds that
$r(t_1, t_3, y) = r_2(t_2, t_3, r_1((t_1, t_2), y))$;
here, we use $t_2$ to denote $A_U(t_1, t_3)$.

We verify that $(U, r)$ is a slice reduction from
$Q_1[S_1]$ to $Q_3[S_3]$.
For each $t_1 \in S_1$, there exists $t_2 \in S_2$ such that
$(t_1, t_2) \in U_1$; and,
for each $t_2 \in S_2$, there exists $t_3 \in S_3$ such that
$(t_2, t_3) \in U_2$.  
Hence, for each $t_1 \in S_1$, there exists $t_3 \in S_3$ such that
$(t_1, t_3) \in U$.
This confirms the coverage condition; we now check correctness.
Suppose that $(t_1, t_3) \in U$, and let $y \in \Sigma^*$.
Set $y' = r_1(t_1, t_2, y)$.
We have that
$$(t_1, y) \in Q_1 \Leftrightarrow (t_2, y') \in Q_2$$
and
$$(t_2, y') \in Q_2 \Leftrightarrow (t_3, r_2(t_2, t_3, y')) \in
Q_3.$$
It follows that 
$$(t_1, y) \in Q_1 \Leftrightarrow (t_3, r(t_1, t_3, y)) \in Q_3.$$

It remains to verify that the function $r$ is FPT-computable with
respect to the parameterization $(\pi_1, \pi_2)$.
We verify this by considering the natural algorithm at this point,
namely, the following algorithm:
given $(t_1, t_3, y)$,
invoke $A_U$ on $(t_1, t_3)$ and, if this computation halts, set
$t_2$ to the result;
then, compute $y' = r_1(t_1, t_2, y)$ using the algorithm 
witnessing FPT-computability, and finally, 
compute $r_2(t_2, t_3, y')$ using the algorithm witnessing
FPT-computability.
On an input $x = (t_1, t_3, y) \in U \times \Sigma^*$,
 the running time of this algorithm can be upper
bounded by
$$F(t_1, t_3) + 
f_1(t_1, t_2) p_1(|(t_1, t_2, y)|) + 
f_2(t_2, t_3) p_2(|(t_2, t_3, y')|)$$
where $F(t_1, t_3)$ denotes the running time of $A_U$ on input 
$(t_1, t_3)$;
 and $(f_1, p_1)$ and $(f_2, p_2)$
witness the FPT-computability of $r_1$ and $r_2$, respectively.
We need to show that this running time is 
degree-bounded with respect
to $(\pi_1, \pi_2)$; 
for the remainder of the proof, let us simply
use \emph{degree-bounded} to mean
degree-bounded with respect
to this parameterization.
We view the various quantities under discussion as functions
of $x = (t_1, t_3, y)$.
Since $F(t_1, t_3)$, $f_1(t_1, t_2)$ and $f_2(t_2, t_3)$ 
can be viewed as 
computable functions of
$(t_1, t_3)$, by Proposition~\ref{prop:db-closure}
(\ref{closure:sum}, \ref{closure:product}),
it suffices to verify that 
each of $p_1(|(t_1, t_2, y)|)$, 
$p_2(|(t_2, t_3, y')|)$ 
is degree-bounded.
By appeal to Proposition~\ref{prop:db-closure}(\ref{closure:polynomial}),
to verify that $p_1(|(t_1, t_2, y)|)$
is degree-bounded,
it suffices to observe that 
each of $|t_1|$, $|t_2|$, $|y|$ is degree-bounded.
Similarly, 
to verify that $p_2(|(t_2, t_3, y')|)$  is degree-bounded,
it suffices to verify that each of
$|t_2|$, $|t_3|$, and $|y'|$ are degree-bounded;
this is clear for $|t_2|$ and $|t_3|$, and the size $|y'|$ is bounded
above by $f_1(t_1, t_2) p_1(|(t_1, t_2, y)|)$, which is degree-bounded
since we just verified that $p_1(|(t_1, t_2, y)|)$ is degree-bounded.
\end{proof}

\section{Discussion: Derivation of the Classification of Prefixed Graphs}
\label{sect:discussion}

Let us say that a sentence is \emph{quantified conjunctive} 
if conjunction ($\wedge$) is the only connective that occurs therein.
A \emph{prefixed graph} is a pair $(P, G)$ where $P$ is a
quantifier prefix and $G$ is a graph whose vertices are the variables
appearing in $P$.
Let $\fancyg$ be a set of prefixed graphs.
Let $\qc(\fancyg)$ denote
the set that contains a prenex quantified conjunctive sentence $P \phi$ if 
there exists a prefixed graph $(P, G) \in \fancyg$
such that
$\phi$ is a conjunction
of atoms, where the variable set of each atom forms a clique in $G$.
We will assume that no variable occurs more than once in an atom;
we make this assumption without loss of interestingness,
since given a sentence $\Phi \in \qc(\fancyg)$ and a structure $\relb$,
one can efficiently compute a sentence $\Phi' \in \qc(\fancyg)$ 
and a structure $\relb'$ such that (1) each atom of $\Phi$
containing more than one variable occurrence is replaced, in $\Phi'$,
with an atom with the same variables but not having multiple variable
occurrences,
and (2) $\relb \models \Phi$ iff $\relb' \models \Phi'$.

The previous work~\cite{ChenDalmau12-decomposingquantified} studied the complexity of model checking
on sentence sets $\qc(\fancyg)$, proving a comprehensive
classification.
Here, we show how this classification can be readily derived using the
main theorem of the present article.
This witnesses the strength and generality of our main theorem
(Theorem~\ref{thm:main}).
We first present two lemmas, then proceed to the derivation.

When $\fancyg$ is a set of prefixed graphs,
let $\normqc(\fancyg)$ denote
the subset of $\qc(\fancyg)$ that contains a sentence if
\begin{itemize}

\item in each atom, the latest occurring variable (in the quantifier prefix)
is existentially quantified, and

\item it is symbol-loose.

\end{itemize}

\begin{lemma}
Let $\fancyg$ be any set of prefixed graphs.
The case problems $\mc[\qc(\fancyg)]$ and
$\mc[\normqc(\fancyg)]$ slice reduce to each other.
\end{lemma}

\begin{proof}
It is straightforward to verify that $\mc[\normqc(\fancyg)]$
slice reduces to $\mc[\qc(\fancyg)]$
(by definition, $\normqc(\fancyg)$ is a subset of $\qc(\fancyg)$).
We slice reduce from $\mc[\qc(\fancyg)]$
to $\mc[\normqc(\fancyg)]$ as follows.
Define $U$ to contain a pair $(\phi, \phi')$ 
of prenex quantified conjunctive sentences
if they share the same quantifier prefix,
$\phi'$ is symbol-loose,
and each atom $\alpha$ of $\phi$ can be placed in bijective
correspondence
with an atom $\alpha'$ of $\phi'$ in such a way that
the variables of $\alpha'$ 
is the set obtained by taking the variables of $\alpha$
and iteratively eliminating the latest occurring variable when it is
universally quantified, 
until the latest occurring variable is existentially quantified.
The partial function $r$ is defined in a natural way, namely,
such that $r(\phi, \phi', \relb)$ is a structure $\relb'$ having the
same
universe $B$ as $\relb$ and having the following property:
for each atom $\alpha = R(u_1, \ldots, u_m)$ of $\phi$, it holds that
the satisfying assignments for $\alpha'$ over $\relb'$
are precisely the satisfying assignments for $\forall D \alpha$ over
$\relb$,
where $D$ denotes the variables in $\alpha$ that do not occur in $\alpha'$.
\end{proof}

\begin{lemma}
Let $\Phi$ be a set of symbol-loose quantified conjunctive sentences.
It holds that $\mc[\Phi]$ and $\mc[\Phi']$ are slice reducible to each
other,
where $\Phi'$ is the graph-like closure of $\Phi$.
\end{lemma}

\begin{proof}
It is straightforward to verify that $\mc[\Phi]$ slice reduces to 
$\mc[\Phi']$, as $\Phi \subseteq \Phi'$.
We thus show that $\mc[\Phi']$ slice reduces to $\mc[\Phi]$.
Let $U$ be the set of pairs $(\phi', \phi)$ such that
$\phi$ is a symbol-loose quantified conjunctive sentence,
and $\phi'$ is in the graph-like closure of $\phi$.
We use $\sigma$ and $\sigma'$ to denote the signatures of
$\phi$ and $\phi'$, respectively.
Since disjunction and negation do not occur in $\phi$,
we have that $\phi'$ is obtained from $\phi$ via applications
of the syntactic transformations $(\alpha)$, $(\beta)$, and
$(\gamma)$,
and replacement.
Since each of these syntactic transformations preserves
logical equivalence as well as the number of atom occurrences,
there is a mapping $S: \sigma \to \sigma'$
such that when each symbol in $\phi$ is mapped under $S$,
the resulting sentence is logically equivalent to $\phi'$.
The partial function $r$ is defined by
$r(\phi', \phi, \relb') = \relb$ where
for each symbol $R \in \sigma$, the relation
$R^{\relb}$ is defined as $S(R)^{\relb'}$.
\end{proof}

We now explain how to obtain the main dichotomy of
the previous work~\cite{ChenDalmau12-decomposingquantified}, in particular,
we show how to classify precisely the sets of prefixed graphs
$\fancyg$
such that $\mc[\qc(\fancyg)]$ is in $\case{\fpt}$.

First, consider the case where atoms in $\normqc(\fancyg)$
may have unboundedly many variables.
If for each $k \geq 1$ there exists an atom in $\normqc(\fancyg)$
with at least $k$ existentially quantified variables, then
there is a direct reduction from $\caseclique$
and one has $\case{\wone}$-hardness of $\mc[\qc(\fancyg)]$.
Otherwise, define
$F_k$ to be the sentence 
$\forall y_1 \ldots \forall y_k \exists x \bigwedge_{i=1}^k
E_i(y_i,x)$;
up to the insertion of additional variables that do not appear in
atoms and up to renaming of variables,
we have that the sentences $F_k$ are instances of $\normqc(\fancyg)$.
By the above two lemmas, it suffices to prove that the 
graph-like closure of $\{ F_k \}$ is hard.
It is readily verified that $\lay( F_k )$
is $\forall \{ y_1, \ldots, y_k \} (\exists x \bigwedge_{i=1}^k
E_i(y_i, x))$, where the formula in parentheses is viewed as 
the single disjunct of a disjunction; we have that $\thick(F_k) = k+1$,
and by the main theorem, we obtain hardness
(either
$\case{\wone}$-hardness or $\case{\cowone}$-hardness)
of the graph-like closure of $\{ F_k \}$.

Next, consider the case 
where there is a constant upper bound on the number of variables
that occur in 
atoms in $\normqc(\fancyg)$.
By our assumption that no variable occurs more than once in an atom,
the set of sentences $\normqc(\fancyg)$ has bounded arity.
In this case, by the two presented lemmas,
we have that $\mc[\qc(\fancyg)]$ is equivalent, under slice reduction,
to $\mc[\Phi']$, where $\Phi'$ is the graph-like closure of
$\normqc(\fancyg)$.  
Since $\normqc(\fancyg)$ has bounded arity, $\Phi'$ does as well.
Hence,
our main theorem
(Theorem~\ref{thm:main}) then can be applied to infer that
$\mc[\Phi']$ is either in $\case{\fpt}$,
is
$\case{\wone}$-hard, or is $\case{\cowone}$-hard.

By arguing as in 
Corollaries~\ref{cor:main-computable}
and~\ref{cor:main-noncomputable}, 
one can derive dichotomies in the complexity of
$\param{\mc[\qc(\fancyg)]}$.

\begin{acks}
The author was supported 
by the Spanish Project FORMALISM (TIN2007-66523), 
by the Basque Government Project S-PE12UN050(SAI12/219), and 
by the University of the Basque Country under grant UFI11/45.
The author thanks Montserrat Hermo and Simone Bova for
useful comments.
\end{acks}

\bibliographystyle{ACM-Reference-Format-Journals}

\bibliography{/Users/hubiec/Dropbox/active/writing/hubiebib.bib}


\end{document}